\documentclass[11pt]{article}
\usepackage[left=1.2in,top=1.2in,right=1.2in,nohead,bottom=1.2in]{geometry}
\usepackage{graphicx}
\usepackage{amsmath, amssymb, amsfonts, amsthm, float,setspace}  
\usepackage{enumerate, color, framed, float, multirow}
\usepackage{comment, longtable, caption, subcaption, appendix}
\usepackage[round]{natbib}
\usepackage[titles]{tocloft}
\usepackage{setspace, parskip}

\allowdisplaybreaks
\newcommand{\ds}{\displaystyle}
\newcommand{\E}{\text{E}}
\newcommand{\Var}{\text{Var}}
\newcommand{\X}{\mathcal{X}}

\newcommand\numberthis{\addtocounter{equation}{1}\tag{\theequation}}

\newtheorem{theorem}{Theorem}

\newtheorem{proposition}{Proposition}
\newtheorem{lemma}{Lemma}
\newtheorem{corollary}{Corollary}

\theoremstyle{remark}

\newtheorem{remark}{Remark}
\newtheorem{assumption}{Assumption}

\author{
Dootika Vats\\
Department of Mathematics and Statistics\\
Indian Institute of Technology Kanpur\\
Kanpur, India 208016\\
\texttt{dootika@iitk.ac.in}
\and
James M. Flegal \\
Department of Statistics\\
University of California\\
Riverside, CA 92521\\
\texttt{jflegal@ucr.edu}
}
\title{Lugsail lag windows for estimating time-average covariance matrices}
\date{\today}

\begin{document}

\maketitle
 
\onehalfspacing
\begin{abstract}
Lag windows are commonly used in time series, econometrics, steady-state simulation, and Markov chain Monte Carlo to estimate time-average covariance matrices. In the presence of positive correlation of the underlying process, estimators of this matrix almost always exhibit significant negative bias, leading to undesirable finite-sample properties.   We propose a new family of lag windows specifically designed to improve finite-sample performance by offsetting this negative bias. Any existing lag window can be adapted into a lugsail equivalent with no additional assumptions. We use these lag windows within spectral variance estimators and demonstrate its advantages in a linear regression model with autocorrelated and heteroskedastic residuals. We further employ the lugsail lag windows in weighted batch means estimators due to their computational efficiency on large simulation output.  We obtain bias and variance results for these multivariate estimators  and significantly weaken the mixing condition on the process. Superior finite-sample properties are illustrated in a vector autoregressive process and a Bayesian logistic regression model. 
\end{abstract}
\section{Introduction} 
\label{sec:introduction}

Variance of estimators in correlated data problems often take the form $\Sigma: = \sum_{s=-\infty}^{\infty} R(s)$, where $R(s)$ is a lag-$s$ covariance matrix. In time series, $\Sigma$ occurs in the estimation of spectra and long run variance \citep{hannan:1970,priest:1981} while in econometrics it occurs in heteroskedastic and autocorrelation consistent (HAC) covariance matrix estimation \citep{andr:1991, newey:west:1987}.  In steady-state simulation and Markov chain Monte Carlo (MCMC), $\Sigma$ is the limiting covariance of Monte Carlo estimators, sometimes referred to as the time-average covariance matrix \citep{glyn:whit:1992, chan:yau:2017}. 

Estimators of $\Sigma$ often downweight the sample lag covariances through a lag window (or kernel function). Such estimators suffer from two sources of bias \citep[see e.g.][]{dehaan:1997}.  A first-order bias term originates from the choice of lag window and is typically $O(b^{-q})$ for some $q \geq 1$ and tuning factor $b$.  A $o(b^{-q})$ second-order bias term is a consequence of finite sampling.  Both terms are typically negative, inducing significant downward bias in the estimation of $\Sigma$. In the univariate case, \citet[Table 1]{chan:2017} summarize the bias of various estimators of $\Sigma$. Of the 17 estimators considered, 15 exhibit negative first-order bias under positive correlation and 2 have a first-order bias of zero. 
Offsetting this negative bias is  imperative, especially in the presence of high positive correlation, a scenario common in steady-state simulation and MCMC.  Here, estimators of $\Sigma$ are critical to determining stopping time of the simulation \citep{glyn:whit:1992} as negatively biased estimators lead to premature simulation termination  and under-coverage of confidence regions. 

We propose a novel and flexible class of lag windows that can offset both the first and second-order bias, while preserving asymptotic unbiasedness and consistency. We call these, lugsail lag windows due to visual similarity with a lugsail (a fore-and-aft, four-cornered sail that is suspended from a spar or yard). The distinguishing feature of lugsail lag windows that allows this offsetting, is that they can take values above 1.  We are unaware of any other lag window with this property and in fact, \cite{berg:pol:2009} claim that ``there is no benefit to allowing the window to have values larger than 1''. 

All commonly encountered lag windows can be easily transformed into a corresponding lugsail lag window with a zero (or even positive) first-order bias. Figure~\ref{fig:all_lags} illustrates this flexibility for three popular lag windows.  The lugsail lag windows can be tuned based on the correlation (or persistence) of the underlying process.  We focus on positive correlation and note the proposed settings may be counter productive in anti-persistent applications.  We quantify the correlation as being \textit{moderate} when it is similar to an AR(1) process with $\rho \in (0,0.7)$.  For moderate correlation, common in the analysis of time series spectra and HAC estimation \citep{laz:lewis:stock}, we recommend the zero lugsail (solid blue line in Figure~\ref{fig:all_lags}) where the first-order bias is zero. 
\begin{figure}[htbp]
	\centering
	\includegraphics[width = 2in]{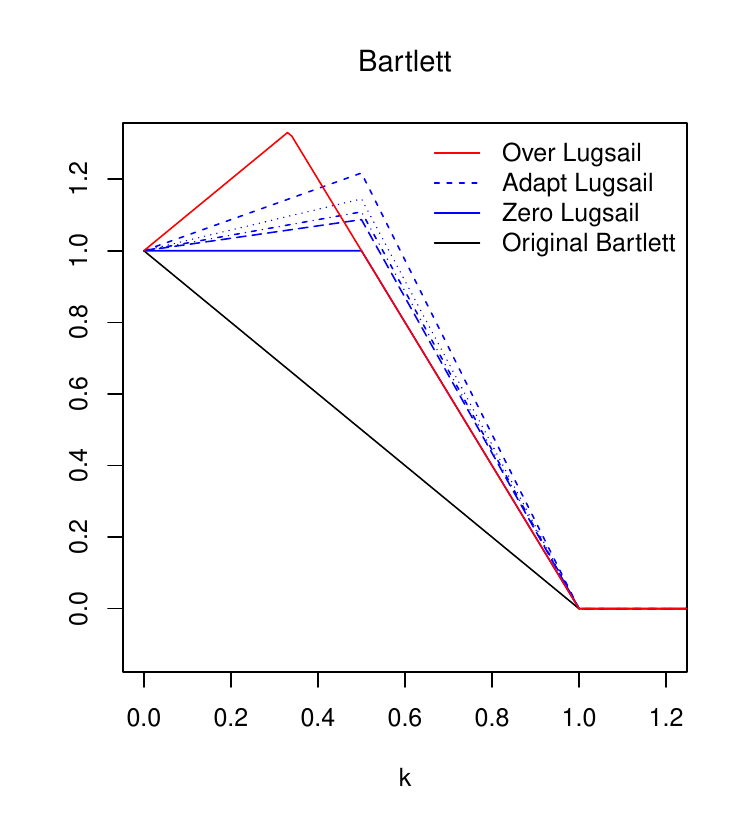} \hspace{-.5cm}
	\includegraphics[width = 2in]{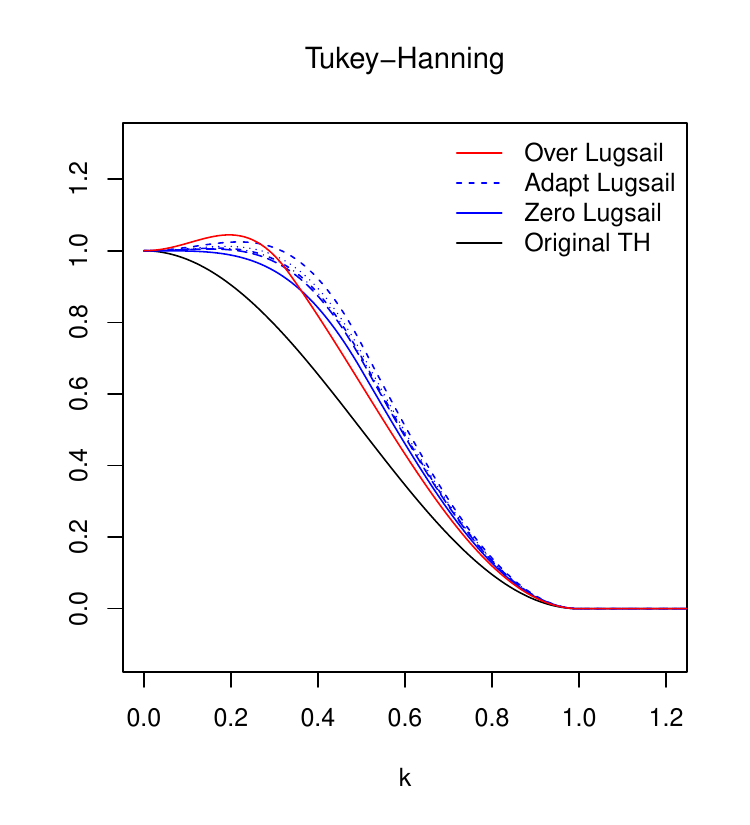}\hspace{-.5cm}
	\includegraphics[width = 2in]{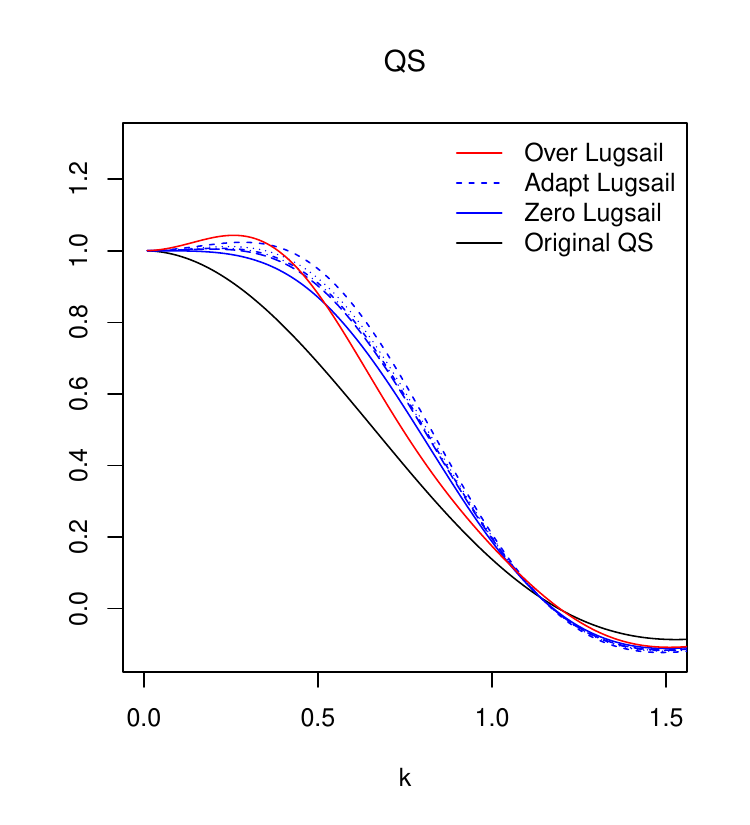}
	\caption{Lugsail versions of Bartlett, Tukey-Hanning, and quadratic spectral lag windows.  Original lag windows are black, zero lugsail lag windows with a first-order bias of zero are blue (adapted versions are dashed blue), and over biased lugsail lag windows are red.}
	\label{fig:all_lags}
\end{figure}
Higher correlation requires more aggressive settings to offset some (or all) of the second-order bias.  For \textit{high} correlations similar to an AR(1) with $\rho \in (0.7, 0.95)$, we recommend using adapt lugsail lag windows which converge to the zero lugsail as data increases. These adapt lugsail windows can also be used in the presence of moderate correlation.  \textit{Extreme} correlation similar to an AR(1) with $\rho \in (0.95, 1)$ is common in MCMC and steady state simulations.  Here, the second-order bias remains significant in finite sampling even for large sample sizes.  In extreme correlation settings, we illustrate the utility of over lugsail lag windows with a positive $O(b^{-1})$ first-order bias, which are appealing when simulation allows for additional data.

We first use lugsail lag windows in spectral variance (SV) estimators, which are the default choice in the analysis of time series spectra and HAC estimation. Univariate and multivariate SV estimators have been discussed in \cite{ande:1971}, \cite{andr:1991}, \cite{hannan:1970}, \cite{newey:west:1987}, \cite{parzen:1957}, \cite{white:1980}, and others in the context of stationary time series, ordinary least squares, generalized method of moments, and instrumental variables. We illustrate, in an HAC example, how lugsail lag windows safeguard against oversized tests. For application in HAC estimation, where data is limited and correlation is moderate, we recommend the zero lugsail lag windows.

   Efforts in improving the finite-sample properties of SV estimators are continual. These include the introduction of new lag windows \citep{kief:vogel:2002a,phill:sun:2006} and  tuning $b$ \citep{kief:vogel:2002b,kief:vogel:2005,sun:2008,wilh:2015}. These recommendations can easily be used in conjunction with lugsail lag windows.  

SV estimators are agonizingly slow for large data and are rarely used in steady-state simulation and MCMC where large simulation lengths are standard. Other conservative estimators of $\Sigma$  proposed by \cite{dai:jone:2017}, \cite{geye:1992}, and \cite{koso:2000} require time-reversibility, retain significant asymptotic bias by design, and are also computationally intensive. 

Weighted batch means (BM) estimators \citep{liu:fleg:2018} provide a fast alternative when used with piece-wise linear lag windows.  We focus on piece-wise linear lugsail lag windows that  yield computationally efficient weighted BM estimators (left plot in Figure~\ref{fig:all_lags}).  As part of our study of bias and variance of weighted BM estimators, we significantly weaken the sufficient conditions.  Previous bias results assume the underlying process to be $\phi$-mixing and require 12th order moments \citep{chie:gold:mela:1997,song:schm:1995}. We require  only $\alpha$-mixing and $4$th order moments. This, for example, allows our results to be applicable to polynomially ergodic Markov chains, rather than only uniformly ergodic Markov chains. 

For a fixed $b$, lugsail estimators may increase asymptotic mean-squared-error compared to the original counterpart.  However, mean-squared-error is an ineffective quality measure for variance estimators as bluntly argued by \cite{simon:1993}, ``it does not address the relative importance of bias and variability, and the differing effects of negative bias and positive bias, on test size and confidence interval''. Essentially, when the goal is to construct confidence regions or do hypothesis tests, it is beneficial to correct for negative bias at the cost of variability. This is evident in Section~\ref{sec:examples}, where lugsail estimators dramatically improve coverage probabilities of confidence regions, estimates of effective sample size, and estimation accuracy of $\Sigma$ in time series and Bayesian logistic regression examples.

\section{Lugsail lag windows} 
\label{sec:lugsail_window}

Let $\{Y_t\}$ be a $p$-dimensional covariance stationary stochastic process with mean $\mu$ and autocovariance $R(s) = \E[(Y_t - \mu)(Y_{t+s} - \mu)^T]$. 
Recall we are interested in estimating the time-average covariance matrix $\Sigma = \sum_{s=-\infty}^{\infty} R(s)$ using the observed process. If $\mu$ is unknown, it is estimated by $\bar{Y} = n^{-1} \sum_t Y_t$ and the lag-$s$ sample covariance matrix is,
\begin{equation}
\label{eq:sample_lag_covariance}
\hat{R}(s) = \dfrac{1}{n} \ds \sum_{t=1}^{n-s}\left(Y_t - \bar{Y} \right)\left(Y_{t+s} - \bar{Y} \right)^T\,.
\end{equation}
If $\mu$ is known, as in \cite{ande:1971}, \cite{hannan:1970}, and \cite{priest:1981}, $\mu$ replaces $\bar{Y}$ in \eqref{eq:sample_lag_covariance}. Estimators of $\Sigma$ weight the sample lag covariances using a lag window function, $k: \mathbb{R} \to \mathbb{R}$  such that $k(0) = 1, k(x) = k(-x)$ for all $x \in \mathbb{R}$. Three common lag windows are
\begin{align*}
\text{Bartlett: } \quad k(x) = & \begin{cases}
	\left(1 - |x|\right), & |x| \leq 1 \\  0, & |x| \geq 1
\end{cases}\\ 
\text{Tukey-Hanning (TH): } \quad k(x) = & \begin{cases}
	\dfrac{1}{2} + \dfrac{1}{2} \cos(\pi x), & |x| \leq 1 \\  0, & |x| \geq 1
\end{cases}\\
\text{Quadratic Spectral (QS): } \quad  k(x) = & \dfrac{25 }{12 \pi^2 x^2} \left(\dfrac{\sin \left(6 \pi x /5 \right) }{6 \pi x /5}  - \cos\left( 6 \pi x /5 \right)\right) \;.
\end{align*}

Figure~\ref{fig:all_lags} plots the lag windows with solid black lines.
All three lag windows are decreasing which leads to downward biased estimation of $\Sigma$.  Lugsail adjustments of these lag windows intentionally lift them over 1 to correct for this. For $r \geq 1$ and a sequence $ c_n\in [0,1)$ such that $c_n \to c$ as $n \to \infty$, define a family of lugsail windows of any existing lag window to be
\begin{equation}
\label{eq:lugsail}
k_L(x) = \dfrac{1}{1-c_n} k(x) - \dfrac{c_n}{1-c_n} k(rx)\,.
\end{equation}
Setting $c_n = 0$ or $r = 1$, yields the original lag window and increasing $r$ increases the lift in the lugsail lag window.  Setting $r = 1/c_n$ with the Bartlett lag window gives the flat-top Bartlett lag window \citep{pol:roma:1995,pol:rom:1996}.  Figure~\ref{fig:all_lags} presents lugsail versions of the Bartlett, TH, and QS lag windows, with the zero and over lugsail presented in solid blue and red, respectively.  The adapt lugsail (dashed blue in Figure~\ref{fig:all_lags}) shows a sequence of lugsail lag windows converging to the zero lugsail lag window as $n$ increases.  The simple and novel overweighting of the initial lag covariances using  lugsail lag windows offsets most of the second-order bias in moderate correlation applications. Section~\ref{sec:practical_considerations} provides practical guidance on selecting of $r$ and $c_n$ with formal definitions of zero, over, and adapt lugsail. 

In the following sections, we use lugsail lag windows in SV and weighted BM estimators. The following umbrella assumption on the stochastic process is made to that ensure $\Sigma$ is finite. 
For a stationary stochastic process $S = \{S_{n}\}$ on a probability space $(\Omega, {\mathcal F}, P)$, set ${\mathcal F}_{s}^{l} = \sigma(S_{s}, \ldots, S_{l})$.  Define the $\alpha$-mixing coefficients for $n=1, 2, \ldots$ as
\[ 
\alpha(n) = \sup_{s \ge 1} \sup_{A \in {\mathcal F}_{1}^{s}, \, B \in {\mathcal F}_{s+n}^{\infty}} | P(A \cap B) - P(A) P(B) | \; .
\]
The process is $\alpha$-mixing if $\alpha(n) \to 0$ as $n \to \infty$. Let $\|\cdot\|$ denote Euclidean norm.  
\begin{assumption}
\label{ass:alpha_mixing}
For some $\delta > 0$ and $q \geq 1$, $\E\|Y_1\|^{2+\delta} < \infty$ and there exists $\epsilon > 0$ such that $\{X_t\}$ is $\alpha$-mixing with $\alpha(n) = o\left(n^{-(q + 1+ \epsilon)(1+2/\delta)} \right)$.
\end{assumption}

\section{Spectral variance estimators}
Let $b \in \mathbb{N}$ be a truncation point (bandwidth), then the multivariate SV estimator is
\begin{equation}
\label{eq:sve}
\dot{\Sigma}_{k,b} = \sum_{s= -(n - 1)}^{n-1} k\left( \dfrac{s}{b}\right) \hat{R}(s)\,.
\end{equation}
Let $\dot{\Sigma}_{k,b}$ and $\dot{\Sigma}_{k, b/r}$ be SV estimators with integer truncation points $b$ and $b/r$, respectively.  Using a lugsail lag window in \eqref{eq:lugsail} with the multivariate SV estimator in \eqref{eq:sve} yields 
\begin{equation} \label{eq:sv.linear}
\dot{\Sigma}_{k,L} =  \dfrac{1}{1-c_n}\dot{\Sigma}_{k,b} - \dfrac{c_n}{1-c_n}\dot{\Sigma}_{k, b/r }\, .
\end{equation}
That is, the lugsail SV estimator, $\dot{\Sigma}_{k,L}$, is a linear combination of SV estimators. By \eqref{eq:sv.linear} and since $c_n \to c$ as $n \to \infty$, lugsail SV estimators retain consistency from the original SV estimators. Sufficient conditions for strong consistency can be found in \cite{deJong:2000} for applications in econometrics and in \cite{vats:fleg:jones:2018} for time-average covariance matrix estimation.
%
\begin{theorem} The lugsail SV estimator $\dot{\Sigma}_{k,L}$ inherits (strong) consistency from $\dot{\Sigma}_{k,b}$.
\end{theorem}
%
%
%
Studying the bias of $\dot{\Sigma}_{k,L}$ requires additional notation.  
%
For $q \geq 1$, a key object is
\[
\Gamma^{(q)} := -\sum_{s=1}^{\infty} s^q\left[R(s) + R(s)^T \right]\,.
\]
Let $\Gamma:= \Gamma^{(1)}$ and denote the $ij$th element of $\Sigma$, $\dot{\Sigma}_{k,b}$, and $\Gamma^{(q)}$ as $\Sigma_{ij}$, $\dot{\Sigma}_{k,b}^{ij}$, and $\Gamma_{ij}^{(q)}$, respectively.

\begin{theorem}
	\label{thm:bias_andrews}
Let $\E\|Y_1\|^{4 + \delta} < \infty$ for some $\delta > 0$ and let $k(x)$ be continuous and uniformly bounded. Further, let Assumption~\ref{ass:alpha_mixing} hold for $q$ such that 
\[
\lim_{x \to 0} \dfrac{1 - k(x)}{|x|^q} = k_q < \infty\,.
\]
If $b^{q+1}/n \to 0$ as $n \to \infty$, then
\[
\text{Bias}(\dot{\Sigma}_{k,b}) =  \dfrac{k_q}{b^q} \Gamma^{(q)} +o\left(\dfrac{1}{b^q} \right)\quad \text{ and}
\]
\[
\dfrac{n}{b} \text{Var}\left(\dot{\Sigma}_{k,b}^{ij} \right) =  \left[\Sigma_{ii}\Sigma_{jj} + \Sigma_{ij}^2 \right] \int_{-\infty}^{\infty} k^2(x) dx + o(1)\,.
\]
\end{theorem}
\begin{proof}
	The proof follows from \citet[page 280]{hannan:1970} and \citet[Lemma 1]{andr:1991}. \cite{hannan:1970} assume $\mu$ is known, which requires $b^q/n \to \infty$ but a standard argument in \citet[Chapter 9]{ande:1971} shows the result holds for when $\mu$ is replaced by $\bar{Y}$ if we allow $b^{q+1}/n \to \infty$.
\end{proof}

\begin{remark}\label{rem:andrews_HAC}
Under the assumptions of \citet[Theorem 2]{andr:1991}, the conclusion of Theorem~\ref{thm:bias_andrews} holds for HAC matrices; however, \cite{andr:1991} assumes $|k(x)| \leq 1$. A careful study of the proof reveals that $|k(x)| \leq 1$ can be replaced with the assumption $|k(x)| \leq d$ for some $0 < d < \infty$. Lugsail lag windows are bounded by $(1-c_n)^{-1}$.
\end{remark}
The first-order bias term, $k_q \Gamma^{(q)} / b^q$, has negative diagonals for positively correlated processes. Most efforts in reducing the bias have gone into choosing $b$.  For example, \cite{andr:1991} provides optimal choices for $b$, which we use in our simulations.  Whatever the choice of $b$, the resulting first-order bias remains negative leading to oversized tests \citep[see e.g.][]{hart:2018}. Using the lugsail lag window in Theorem~\ref{thm:bias_andrews}, when $r$ and $c_n$ are such that $c_n > 1/r^q$, the resulting first-order bias is positive. 
\begin{corollary}
\label{cor:bias_lugsail_andrews}
Under the conditions of Theorem~\ref{thm:bias_andrews},
\[
\text{Bias}(\dot{\Sigma}_{k,L}) =  \dfrac{1 - c_nr^q}{1 - c_n}\, \dfrac{k_q}{b^q} \Gamma^{(q)} + o\left( \dfrac{1}{b^q}\right)\,, \text{ and }
\]
\[
\dfrac{n}{b} \text{Var}\left(\dot{\Sigma}_{k,L}^{ij} \right) =  \left[\Sigma_{ii}\Sigma_{jj} + \Sigma_{ij}^2 \right] \int_{-\infty}^{\infty} k_L ^2(x) dx + o(1)\,.
\]
\end{corollary}
The online supplement provides bias and variance expressions for some choices of $c_n$ and $r$.

SV estimators are prohibitively expensive when $p$ or $n$ are large, such as in MCMC and steady-state simulations. Specifically, the optimal choice of $b$ is proportional to $\lfloor n^{\nu} \rfloor$ for $\nu > 0$, leading to an expensive summation in \eqref{eq:sve}. Simulations exhibiting this behavior are in Section~\ref{sub:time_series_example}. The following section focuses on a computationally efficient lugsail estimator.

\section{Weighted BM estimators}

\subsection{Notation} 
\label{sub:problem_setup}
We first define standard notation in steady-state simulation and MCMC where batched estimators are common. These differ slightly from Section~\ref{sec:lugsail_window}, but the form of $\Sigma$ remains the same.  Let $\{X_t\}$ be an $F$-stationary process defined on a $d$-dimensional space, $\X$. For a function $g: \X \to \mathbb{R}^p$, let $\{Y_t\} = \{g(X_t)\}$. Interest is in quantifying the error in estimating $\mu_g = \int g(x) F(dx)$ with $\bar{Y} = n^{-1} \sum_{t=1}^{n} Y_t$. Particularly, $\Sigma = \sum_{s = -\infty}^{\infty} R(s) = \lim_{n \to \infty} n\Var_F(\bar{Y})$.

We continue to assume the process satisfies Assumption~\ref{ass:alpha_mixing}. Particularly in MCMC, $F$ is typically the target distribution and if the Markov chain is polynomially ergodic of order $\xi > (1+q+\epsilon)(1 + 2/\delta)$, the mixing condition in  Assumption~\ref{ass:alpha_mixing} is satisfied \cite[see][]{jone:2004}.

\subsection{Lugsail weighted BM}

The weighted BM estimator incorporates a lag window in combination with non-overlapping batches. For $s = 1, \dots, b < n$ and lag window $k$, denote $\Delta_2 (s)=k( (s-1)/b) - 2 k(s/b) + k((s+1)/b)$. Let $a_s = \lfloor n/s\rfloor$ and for $l = 0, \dots, a_s - 1$,  define $\bar{Y}_l(s) = s^{-1} \sum_{t=1}^{s} Y_{ls+t}$. The weighted BM estimator is
\begin{equation}
\label{eq:wbm}
	\hat{\Sigma}_{k,b} = \sum_{s=1}^{b} \dfrac{1}{a_s - 1} \sum_{l=0}^{a_s-1} s^2 \Delta_2 (s) (\bar{Y}_l(s) - \bar{Y})(\bar{Y}_l(s) - \bar{Y})^T .
\end{equation}
Using the lugsail lag window at \eqref{eq:lugsail} in \eqref{eq:wbm} yields
\begin{equation}
\label{eq:lugsail_bm_general}
	\hat{\Sigma}_{k,L} = \dfrac{1}{1-c_n} \hat{\Sigma}_{k,b} - \dfrac{c_n}{1-c_n}\hat{\Sigma}_{k, b/r}\,,
\end{equation}
where $\hat{\Sigma}_{k,b}$ and $\hat{\Sigma}_{k, b/r}$ are weighted BM estimators with respective integer batch sizes $b$ and $b/r $.  

\begin{theorem} The lugsail estimator, $\hat{\Sigma}_{k,L}$ inherits (strong) consistency from $\hat{\Sigma}_{k,b}$. 
\end{theorem}

%

\subsection{Computational efficiency}
In general, weighted BM estimators can have a similar order of computational complexity as SV estimators.  Since computational efficiency is a necessity, we focus on piece-wise linear lag windows for which $\Delta_2 = 0$ almost everywhere eliminating most terms of the outer sum in \eqref{eq:wbm}.  Specifically, using the Bartlett lag window in \eqref{eq:wbm} yields the fast multivariate BM estimator \citep{chen:seila:1987}. For $n = ab$, let $a$ be the number of batches and $b$ be the batch size such that the following standard assumption holds.
\begin{assumption}	\label{ass:b_increasing}
The integer sequence $b$ is such that $b \to \infty$ and $n/b \to \infty$ as $n\to \infty$, and both $b$ and $n/b$ are nondecreasing.
\end{assumption}
For $l = 0, 1, \dots, a-1$, the mean vector for batch $l$ of size $b$ is $\bar{Y}_l(b) = b^{-1} \sum_{t=1}^{b}Y_{lb + t}$. Then the multivariate BM estimator is
\begin{equation*}
	\label{eq:bm}
	\hat{\Sigma}_{b} = \dfrac{b}{a-1} \sum_{l=0}^{a-1} (\bar{Y}_l(b) - \bar{Y})(\bar{Y}_l(b) - \bar{Y})^T\,.
\end{equation*}
%
Using the lugsail Bartlett lag window in \eqref{eq:lugsail_bm_general}, we obtain the lugsail BM estimator
\begin{equation}
\label{eq:lugsail_bm_multi}
	\hat{\Sigma}_{L} = \dfrac{1}{1-c_n} \hat{\Sigma}_{b} - \dfrac{c_n}{1-c_n}\hat{\Sigma}_{b/r }\,.
\end{equation}
%
We present the bias results for $\hat{\Sigma}_{b}$  and $\hat{\Sigma}_{L}$ which indicate that if $r > 1/c_n$, the lugsail BM estimator has a positive first-order bias.  The proof of the following theorem is in the online supplement. 
%

\begin{theorem}
	\label{thm:bias_L}
Under Assumption~\ref{ass:alpha_mixing} with $q = 1$, 	
	\[\text{\text{Bias}} \left(\hat{\Sigma}_{b} \right) =  \dfrac{\Gamma}{b} + o \left(\dfrac{1}{b} \right)\,.\]
As a consequence,
\begin{equation*}
\label{eq:bias_lugBM}
	\text{\text{Bias}}\left(\hat{\Sigma}_{L}\right) =  \dfrac{\Gamma}{b} \left(\dfrac{1 -rc_n}{1-c_n}\right) + o\left( \dfrac{1}{b}\right)\,.
\end{equation*}
Further, under Assumption~\ref{ass:b_increasing}, $\displaystyle \lim_{n\to \infty} \text{\text{Bias}}\left(\hat{\Sigma}_{L}\right) = 0$.
\end{theorem}

Theorem~\ref{thm:bias_L} makes two contributions. First, the bias of the multivariate BM estimator has not been studied as \cite{chie:gold:mela:1997}, \cite{fleg:jone:2010}, and \cite{song:schm:1995} only consider $p=1$.  Second, the results therein assume $\phi$-mixing and 12 finite moments.  This is especially problematic since $\phi$-mixing Markov chains are uniformly ergodic, a property that is often not satisfied for MCMC algorithms. Assumption~\ref{ass:alpha_mixing} significantly weakens these conditions. 

We require $\{Y_t\}$ to satisfy a strong invariance principle to establish variance of the lugsail BM estimator. Let $B(n)$ be a $p$-dimensional standard Brownian motion. 

\begin{theorem} \label{thm:kuelbs}
\citep{kuel:phil:1980} Under Assumption~\ref{ass:alpha_mixing} for $q = 1$, a strong invariance principle holds. That is, there exists a $p\times p$ lower triangular matrix $L$, $\lambda > 0$ with $\psi(n) = n^{1/2 - \lambda}$, a finite random variable $D$, and a sufficiently rich probability space $\Omega$ such that for almost all $\omega \in \Omega$ and for all $ n > n_0$, with probability 1,
\begin{equation}
	\label{eq:SIP}
	\left\|\sum_{t=1}^{n} {Y_t} - n \mu_g  - LB(n) \right \| < D(\omega) \psi(n)\,.
\end{equation}
\end{theorem}
\begin{remark}\label{rem:weak_cond_sip}
In fact, \cite{kuel:phil:1980} require a slightly weaker condition of $\alpha(n) = o(n^{-(1 + \epsilon)(1 + 2/\delta)})$, but in order to keep all our results under the same umbrella assumption, we assume $\alpha(n) = o(n^{-(2 + \epsilon)(1 + 2/\delta)})$. Here $L$ is such that $LL^T = \Sigma$. A strong invariance principle with rate $\psi(n)$ is known to hold for many processes, including regenerative processes, $\phi$-mixing, and strongly mixing processes \citep[see][for a discussion]{vats:fleg:jones:2018}. 
\end{remark}
The proof of the following theorem is in the supplementary material.
\begin{theorem}
	\label{thm:variance}
	Let the assumptions of Theorem~\ref{thm:kuelbs} hold such that $\E D^4 < \infty$ and $\E_F \|Y_1^4\| < \infty$. Further, let $b$ satisfy Assumption~\ref{ass:b_increasing}  and $b^{-1} \psi^2(n) \log n \to 0$ as $n \to \infty$. Then
\[ 
\dfrac{n}{b} \Var\left(\hat{\Sigma}^{ij}_{L}\right) = \left[\dfrac{1}{r}  + \dfrac{r-1}{r(1-c_n)^2}\right] \left(\Sigma_{ij}^2 +  \Sigma_{ii}\Sigma_{jj} \right) + o \left(1 \right) \,.
\]
\end{theorem}
The batch size $b$ is often chosen to be of the form $\lfloor n^{\nu} \rfloor$ for some $0 < \nu \leq 1/2$. Literature on optimal batch sizes for multivariate BM estimators is not as rich as SV estimators, although univariate suggestions exist \citep{damerdji:1995}. In our simulations, we set $b = \lfloor n^{1/2} \rfloor$. 

\section{Practical considerations} 
\label{sec:practical_considerations}

The lugsail family of lag windows are designed to address negative bias in estimators of $\Sigma$. The degree of negative bias depends on the persistence of the correlation in the process, which we denote as moderate, high, and extreme for underlying processes similar to an AR(1) process with coefficient $\rho \in \left\{ [0,0.7), [0.7, 0.95), [0.95, 1) \right\}$, respectively.  Finite-time bias of the variance estimators depends on the persistence of the correlation through $\Gamma^{(q)}$. Thus, one universal recommendation of $r$ and $c_n$ is unreasonable. Since the choice of $r$ and $c_n$ dictates the amount of the positive first-order bias induced, we present correlation-dependent recommendations are summarized in Table~\ref{tab:recos}.

Consider the exact bias of the univariate BM estimator, $\hat{\sigma}^2_b$, \citep{akta:tuba:2007}:
\begin{align*}
\text{Bias}(\hat{\sigma}^2_b)
& = -\dfrac{2(a+1)}{ab}\ds\sum_{s=1}^{b - 1} sR(s) - 2 \sum_{s=b}^{\infty} R(s) -   \dfrac{2}{b-1} \sum_{s = b}^{n-1}\left( 1 - \dfrac{s}{n} \right) R(s)\,.
\end{align*}
Figure~\ref{fig:persis} plots the relative bias (bias divided by truth) of the univariate BM estimator for moderate, high, and extreme correlation autoregressive processes under the different lugsail settings. For moderate correlation, the original BM estimator exhibits small but noticeable  negative bias.  The zero lugsail corrects for most (but not all) of the negative bias while the adapt lugsail demonstrates minimal positive bias and is recommended in such situations.  
\begin{figure}[htbp]
	\centering
	\includegraphics[width = 2in]{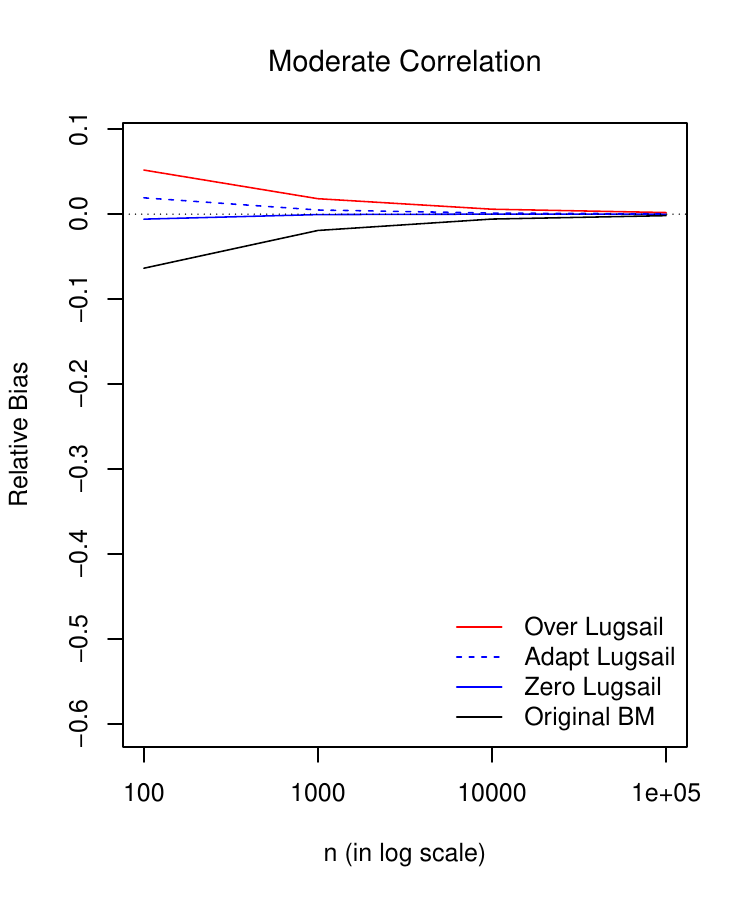} \hspace{-.5cm}
	\includegraphics[width = 2in]{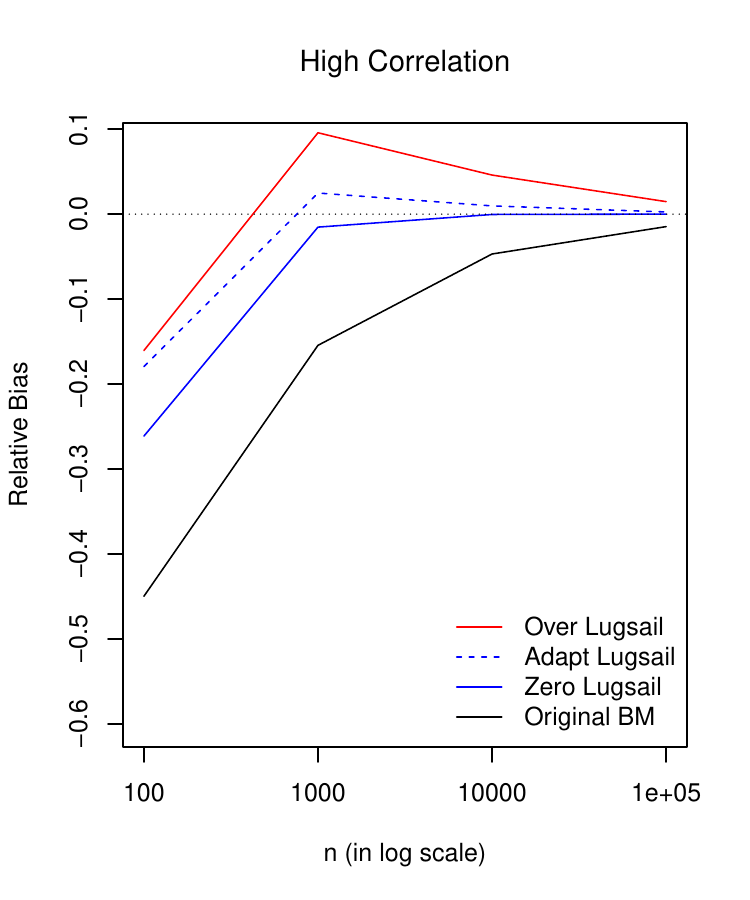}\hspace{-.5cm}
	\includegraphics[width = 2in]{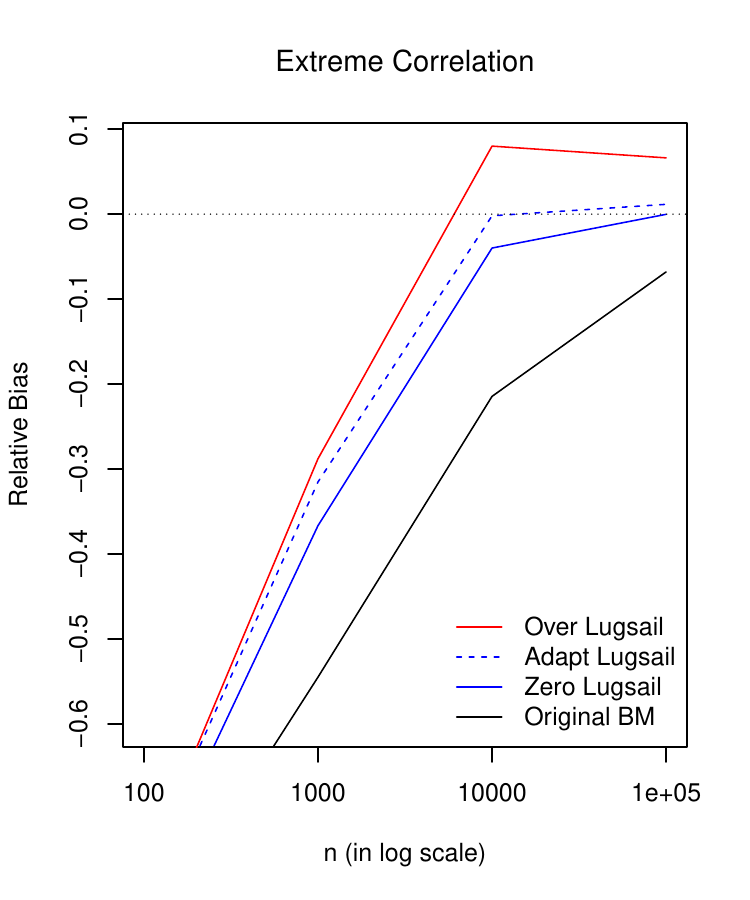}
	\caption{Relative bias of BM versus $\log(n)$ (with $n = 10^2$ to $10^5$) for AR$(p)$ models of different correlations.}
	\label{fig:persis}
\end{figure}

The adapt lugsail performs bias-adjustment as a function of the data, $n$. In small samples, bias is significant, so $c_n$ should be large to allow for a larger offset. For large $n$, the second-order bias is negligible, so $c_n$ and $r$ can be chosen so as to yield zero bias in the first-order term. Thus, we choose $c_n\downarrow c:= r^{-q}$, implying that the adapt lugsail lag window will converge to zero-bias lag windows as $n$ increases. The rate of decay to $c$ should be slow enough to demonstrate this trade-off, which we set as 
\[
c_n^{\text{a}} = \dfrac{\log(n) - \log(b) + 1}{r^q(\log(n) -\log(b)) + 1}\,.
\]
This choice is similar in spirit to the jackknifed estimators of \cite{ding:alex:2015}.  By Assumption~\ref{ass:b_increasing}, $c_n^{\text{a}} \to r^{-q}$, and hence converges to the zero-bias lag window.  We set $r = 2$ ensuring that $r$ is large enough to yield bias adjustment and small enough to control the variance gain. These choices of $c_n$ and $r$ yield the adapt lugsail lag window. In HAC applications where low to moderate correlation is typical, the use of the adapt lugsail is reasonable. 

For high and extreme correlations in Figure~\ref{fig:persis}, the original BM demonstrates substantial negative bias where the adapt and over lugsails are able to remove the negative bias for large sample sizes, albeit with some overestimation. For MCMC, where high and extreme correlation are prevalent (see for example, Section~\ref{sub:bayesian_probit_regression}), a controlled overestimation using over lugsail is often a non-concern as obtaining further samples is relatively easy. Thus, here we recommend setting $r = 3$ and $c^{\text{o}} = 2/(1 + r^q)$, which is based on the fact that it is better to overestimate the variance rather than underestimate it \citep{simon:1993}.  Our choice is based on offsetting the state-of-the-art first-order bias of say, $-m$, since the current literature has been satisfied with this underestimation.  Specifically, we choose $r$ and $c_n$ to induce a first-order bias in the right direction of $+m$, which decreases the overall negative bias considerably.  From our bias results, the over lugsail choices of $r$ and $c_n = c^{\text{o}}$ require 
\[
\dfrac{1 - r^qc^{\text{o}}}{1-c^{\text{o}}} = -1 \Rightarrow c^{\text{o}} = \dfrac{2}{1 + r^q}\,.
\]
We set $r = 3$ so that for $q = 1$, this yields $c = 1/2$ and for $q = 2$, we get $c = 1/5$. Although other choices of $r$ may also be used, smaller choices yield larger $c^{\text{o}}$ increasing the variance of the estimator and larger choices are practically inconvenient as they require $b$ to be large enough so that $b/r$ is large enough. Empirically $r = 3$ provides a good balance between these trade-offs. The online supplement presents variances of the estimators for these choices of $r$ and $c$ along with general expressions. 

An additional practical concern is that finite sample estimates of $\Sigma$ may not be positive-definite. In fact, Bartlett, QS, and BM estimators are only guaranteed to be positive-semidefinite.  Further, TH and lugsail estimators can have negative eigenvalues.  To ensure positive-definiteness, we provide an adjusted estimator that retains the large sample properties of the original similar to \cite{jen:pol:2015}.

\begin{table}[tb]
	\caption{Summary of the recommendations for the lugsail lag window settings.}
	\label{tab:recos}
	\centering

	\begin{tabular}{|l|c|l|l|}
	\hline
        Correlation  & Lugsail window & $r$ & $c_n$ \\ \hline
	Moderate & Zero Lugsail & 2  & $r^{-q}$ \\
	Moderate to High & Adapt Lugsail  & 2 & $c_n^{\text{a}}$ \\ 
	High to Extreme & Over Lugsail & 3 & $c^{\text{o}}$ \\
	\hline
	\end{tabular}
\end{table}

Let $\hat{\Sigma}_n$ be any estimator of $\Sigma$ and $\hat{V} = \text{diag}(\hat{\Sigma}_n)$, that is, $\hat{V}$ is the diagonal matrix of the univariate variance estimates. Consider the correlation matrix corresponding to $\hat{\Sigma}_n$, $\hat{C}_n = \hat{V}^{-1/2} \hat{\Sigma}_n \hat{V}^{-1/2}$. Note $\hat{C}_n$ is a symmetric matrix with real-valued entries, and hence the eigenvalue decomposition $\hat{C}_n = P\hat{D}_n P^T$ exists. Here $P$ is a $p\times p$ orthogonal matrix and $D = \text{diag}(\hat{d}_1, \dots, \hat{d}_p)$ is the diagonal matrix of eigenvalues of $\hat{C}_n$. 

If $\hat{\Sigma}_n$ is not positive-definite, some eigenvalues, $\hat{d}_i$, are not positive. To correct define $\hat{d}^+_i = \max\{\hat{d}_i, \epsilon n^{-u}\}$ for $\epsilon > 0$ and $u > 0$. Then $\epsilon n^{-u} \to 0$ as $n\to \infty$ and $\lim_{n\to \infty}\hat{d_i} = d_i > 0$ due to positive-definiteness of $\Sigma$; here $d_i$ are the eigenvalues for the population correlation matrix. Let $\hat{D}^+ = \text{diag}(\hat{d}^+_1, \dots, \hat{d}^+_p)$, then the adjusted estimator is
\[
	\hat{\Sigma}_n^+ = \hat{V}^{1/2}P^T \hat{D}^+ P\hat{V}^{1/2}\,.
\]
The constants $\epsilon$ and $u$ are user-chosen. We suggest $\epsilon = \sqrt{\log(n)/p}$ and $u = 9/10$, which work well in practice with no problem-specific tuning required. 

\section{Examples} 
\label{sec:examples}

\subsection{HAC estimation example} 
\label{sub:time_series_example}

For $t = 1, 2, \dots, n,$ consider the linear regression model for $y_t = x_t^T \beta + u_t\,,$ where $\beta$ is a $p$-vector of coefficients, $x_t$ is the $p$-vector of covariates, and $u_t$ are autocorrelated, zero mean, and possibly conditionally heteroskedastic. The ordinary least squares estimator of $\beta$ is $\hat{\beta}_{\text{ols}} = \left(\sum x_tx_t^T \right)^{-1} \sum x_ty_t$.
Let $v_t = x_tu_t$, and consider the process $\{v_t\}$. In many situations, the estimator satisfies asymptotic normality, so that as $n \to \infty$,
\[
\sqrt{n}(\hat{\beta}_{\text{ols}} - \beta) \overset{d}{\to}N(0, M \Sigma M)\,,
\]
where $M$ is a known symmetric matrix and $\Sigma = \sum_{k=-\infty}^{\infty} \text{Cov}(v_t, v_{t+k})\,.$
Inference on $\hat{\beta}_{\text{ols}}$ is critically dependent on the estimator of $\Sigma$. Additional assumptions on the process $\{v_t\}$ and the ordinary least squares estimator $\hat{\beta}_{\text{ols}}$ are discussed in \cite{andr:1991}.
 
We implement an AR1-HOMO model which constructs two independent AR(1) processes for the errors $u_t$ and the regressors $x_t$, so that $\{v_t\}$ is a centered process. Consider for $t = 1, \dots, n$ $x_t = \rho_x x_{t-1} + \alpha_t$ and $u_t = \rho_u u_{t-1} + \epsilon_t$,
where $\alpha_t \sim N_p(0, W)$ and $\epsilon_t \sim N(0,w)$. The limiting distribution of $x_t$ is $N_p(0, \Lambda = W/(1-\rho_x^2) )$ and of $u_t$ is $N(0, \lambda = w/(1-\rho_u^2)$.  The true $\Sigma$ is
\[
\Sigma = \lambda \Lambda + \lambda(\Lambda + \Lambda^T) \left(\dfrac{\rho_u \rho_x}{1 - \rho_u \rho_x} \right)\,.
\]
%
We set $p = 5$, $w = 1$, $W$ to be the AR correlation matrix with coefficient $.99$. We let $\rho_u = \rho_x \in \{.50, .70, .90\}$ and generate the process for $n \in \{500, 1000\}$ with $\beta$ set to be the zero vector.

Since the goal is inference for $\hat{\beta}_{\text{ols}}$, we follow \cite{simon:1993} and assess the quality of estimation of $\Sigma$ via coverage probabilities of asymptotic confidence regions. Table~\ref{tab:hac_coverage} contains results from 1000 replicated simulations.  Almost systematically, lugsail lag windows yield higher coverage probabilities than the original lag windows. Since the center of all the regions are the same, the difference in coverage probability is a direct consequence of the lugsail estimators being larger in its determinant than the non-lugsail versions. As expected, the coverage is lowest in the higher correlation instances and when $n$ is small.  Coverage probabilities for the TH and QS lag windows are virtually identical, and hence the TH results are not included in Table~\ref{tab:hac_coverage}.

\begin{table}[tb]
	\caption{HAC: Coverage probabilities for 90\% confidence regions over 1000 replications for lugsail Bartlett and QS lag windows (largest standard error is .0158).}
	\label{tab:hac_coverage}
\centering
	\begin{tabular}{|c|c|cccc|cccc|}
	\hline
	$n$ &  $\rho_u$ & \multicolumn{4}{c|}{Bartlett} & \multicolumn{4}{c|}{Quadratic Spectral} \\  \hline 
	&  & - & Zero & Adapt & Over  & - & Zero & Adapt & Over\\ \hline
500 & .50 & 0.825  &  0.837  & 0.838 &  0.861  & 0.836 &  0.853 &  0.858 &  0.864\\ 
500 & .70  & 0.757  &  0.774  &  0.769  &  0.798  &  0.784  &  0.791  &  0.792  &   0.801 \\ 
500 & .90  & 0.553  & 0.532  & 0.529  &  0.573 & 0.588  & 0.581  & 0.579 &  0.596 \\  \hline
1000 & .50 &  0.840 &  0.851  & 0.852  & 0.869  & 0.851  & 0.855 &  0.857  & 0.861\\ 
1000 & .70  & 0.821  &  0.842  & 0.841 &  0.867  & 0.851  & 0.860  & 0.859  & 0.863 \\ 
1000 & .90  &  0.661  & 0.672  & 0.670  & 0.717  & 0.696  & 0.706  & 0.704  & 0.721 \\  \hline
	\end{tabular}
\end{table}

Since the true value of $\Sigma$ is known, we can also estimate  bias. Over 1000 replications, we record the average relative bias on the diagonals of an estimate $\hat{\Sigma}$,
\[
	\dfrac{1}{p} \ds \sum_{i=1}^{p} \dfrac{\hat{\Sigma}_{ii} - \Sigma_{ii}}{\Sigma_{ii}}\,,
\]
which conserves bias direction. The results for $n = 500$ are in Figure~\ref{fig:hac_bias}. In each setting, the lugsail version of the lag windows exhibit smaller downward bias where this bias is more significant for $\rho_x = \rho_u = .90$. Again the results for TH and QS windows (where $q = 2$) are virtually identical while being less negatively biased compared to the Bartlett family of lag windows.
\begin{figure}[tb]
	\centering
\includegraphics[width = 5.8in]{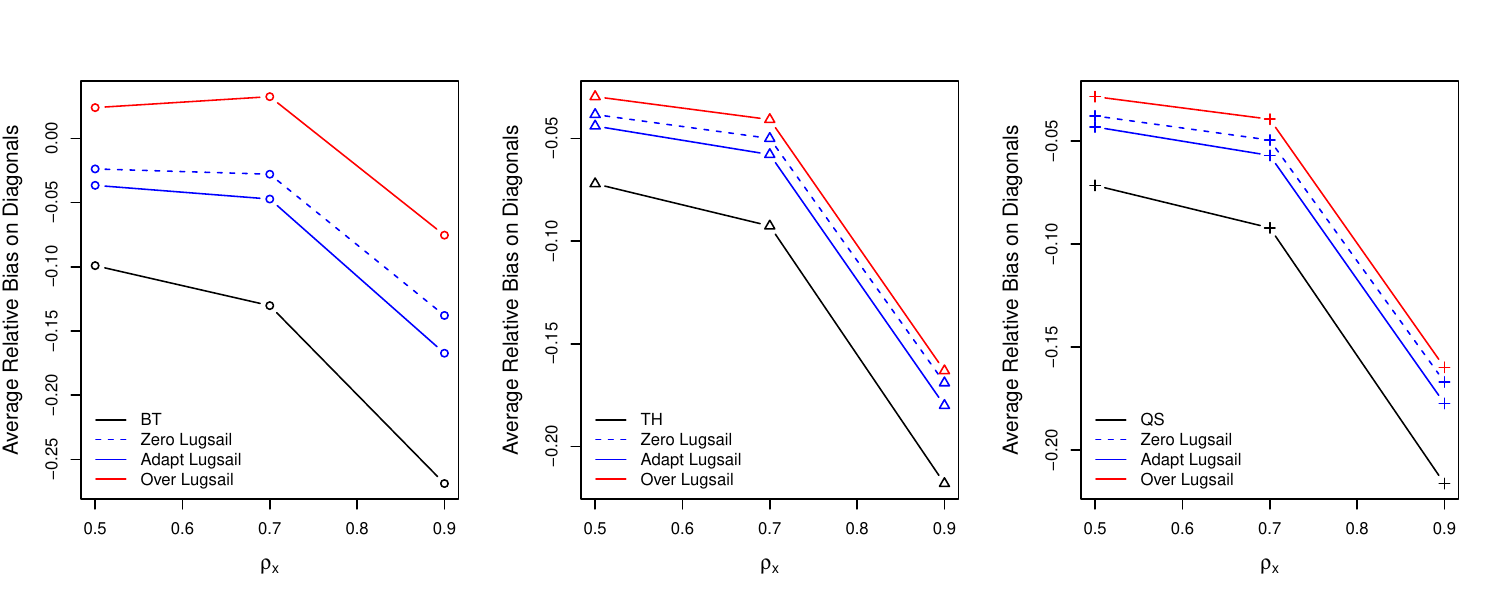}
	\caption{HAC: Average relative bias on the diagonals for Bartlett (left), TH (middle), and QS (right).}
	\label{fig:hac_bias}
\end{figure}

Finally, Table~\ref{tab:sv_time} compares compute time for lugsail BM versus lugsail SV estimators with $\rho_u = .50$, $p \in (10, 50)$ and $n \in (10^3, 10^4, 10^5)$. Output of this size is common in MCMC and steady-state simulations. For a $10$-dimensional problem, SV estimators are orders of magnitude slower than BM estimators, and the QS lag window is particularly slow since the lag window is not truncated. Thus, in the rest of the examples, where $n$ and $p$ will be large, we focus on only weighted BM estimators. 
\begin{table}[bt]
\caption{Running time (sec) for lugsail  BM and lugsail SV estimators from 10 replications.}
\label{tab:sv_time}
\centering
\begin{tabular}{|l|rr|rr|rr|}
\hline
$n$& \multicolumn{2}{c|}{$10^3$} &  \multicolumn{2}{c|}{$10^4$} &  \multicolumn{2}{c|}{$10^5$} \\  \hline
$p$ & $10$ & $50$ & $10$ & $50$ & $10$ & $50$\\ \hline
Over BM & 0.001  & 0.002  &  0.001  &   0.006   &   0.010    &   0.078 \\ 
Over Bartlett & 0.025  & 0.186  &  0.406   &  4.357   &   7.590   &   94.317 \\ 
Over TH & 0.024 & 0.183  &  0.333   &  3.392   &   4.948  &   55.689 \\ 
Over QS & 0.275  & 4.059  & 28.566  & 410.518  & 2225.032  & 31876.747 \\  \hline
\end{tabular}
\end{table}

\subsection{Time series example} 
\label{sub:vector_autoregressive_process}

We consider a vector autoregressive process of order 1. Let $\Phi$ be a $p\times p$ matrix with spectral norm less than 1 and let $\Omega$ be a $p \times p$ positive definite matrix. For $t = 1, 2, \dots$ let $Y_t \in  \mathbb{R}^p$ such that $Y_t = \Phi Y_{t-1} + \epsilon_t$ and $\epsilon_t \sim N_p(0, \Omega)$.  The stationary distribution is $N_p(0, V)$ where $vec(V) = (I_p^2 - \Phi \otimes \Phi)^{-1} vec(\Omega)$. In addition, the chain is geometrically ergodic when the spectral norm of $\Phi$ is less than 1 \citep{tjos:1990} so Assumption~\ref{ass:alpha_mixing} is satisfied. We are interested in estimating $\E Y_t = 0$ and the true  $\Sigma$ is available in closed form \citep{dai:jone:2017}. We set $\Phi = \rho I_p$ for $\rho > 0$ and $\Omega$ to be the AR correlation matrix with coefficient .9. 
%
%
%
%
%
We set $p = 10$ and  $\rho = .95$ and over 1000 replications estimate $\Sigma$ using three different BM methods. The settings are chosen to generate a high correlation process. 

Since $\Sigma$ is known for this example, we estimate the average relative bias on the diagonals over 1000 replications in Figure~\ref{fig:var_bias}. As $n$ increases, the relative bias for all three methods converge to zero, however, over lugsail is converging from above while all others are mostly converging from below.
\begin{figure}[htb]
	\centering
	\includegraphics[width = 2.5in]{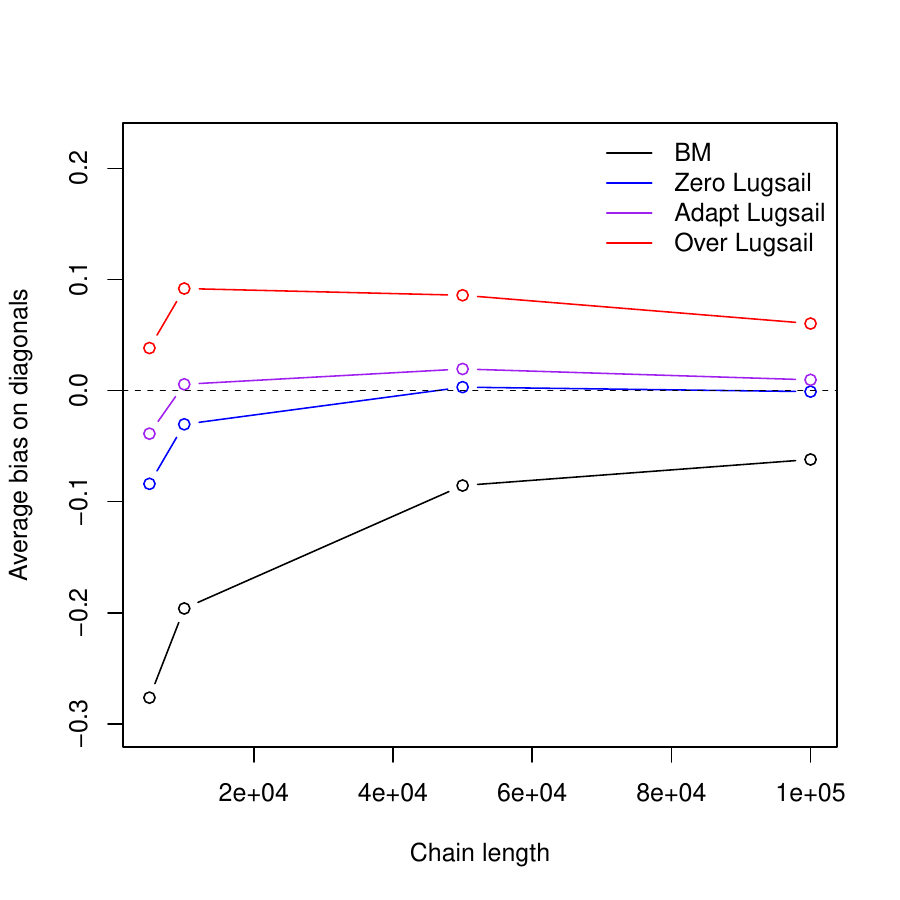}
	\includegraphics[width = 2.5in]{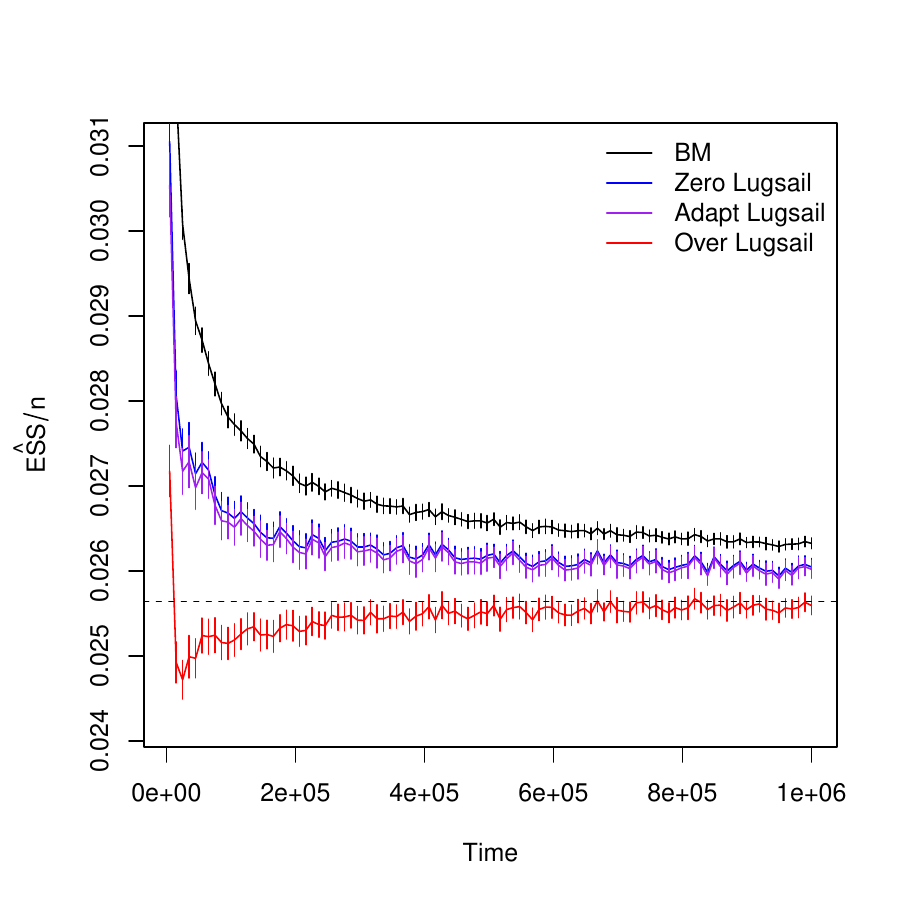}
	\caption{VAR: (Left) Average relative bias on the diagonals over 1000 replications. (Right) Running plot of $\widehat{\text{ESS}}/n$ (with standard errors) from 100 replications.}
	\label{fig:var_bias}
\end{figure}

Let $|\cdot|$ denote determinant. One practical use of $\Sigma$ is determining the effective sample size as a way of summarizing the variability \citep{vats:fleg:jones:2017b}. Since for $n \to \infty$,
\[
\dfrac{\widehat{\text{ESS}}}{n} = \left(\dfrac{|\widehat{\Var}_F(Y_1)|}{|\hat{\Sigma}|} \right)^{1/p} \to \left(\dfrac{|\Var_F(Y_1)|}{|{\Sigma}|} \right)^{1/p}\,,
\]
we compare the quality of estimation of $\widehat{\text{ESS}}/n$. It is critical that $\widehat{\text{ESS}}$ is not overestimated, so as to not cause early termination of an MCMC simulation.  Figure~\ref{fig:var_bias} presents the average of the running estimate of the effective sample size divided by $n$ for BM estimators based on 100 replications. The original BM, zero lugsail, and adapt lugsail all produce significant over estimation of $\widehat{\text{ESS}}/n$ that will lead to early termination. On the other hand, the over lugsail estimate of $\widehat{\text{ESS}}/n$ converges to the truth from below, which dramatically improves the quality of estimation and ensures simulations are not terminated prematurely.

\subsection{Bayesian logistic regression} 
\label{sub:bayesian_probit_regression}

We consider a subset of the data from the ongoing cardiovascular study of the residents of Framingham, Massachusetts (from Kaggle.com with 4238 observations). The  binary response variable identifies whether the patient has a 10 year risk of coronary heart disease. There are 15 covariates including demographic information, behavioral information, medical history, and present medical condition. Since some covariates are categorical, the model matrix for a regression model with intercept has rows $x_i^T = (x_{i1}, \dots, x_{i18})^T$. Let $\beta \in \mathbb{R}^{18}$ and consider a Bayesian logistic regression model with intercept,
\[
\Pr(Y_i = 1) = \dfrac{\exp\{x_i^T \beta \}}{1 + \exp\{x_i^T \beta \}}\,. 
\]
We assign a multivariate normal prior, $\beta \sim N_{18}(0, 100 I_{18})$. We use the \texttt{MCMCpack} library in \texttt{R} to sample from the posterior distribution which runs a random walk Metropolis-Hastings sampler with a normal proposal distribution and consider estimating the posterior mean.


%

Figure~\ref{fig:pro.boxplots} shows a running plot of $\widehat{\text{ESS}}/n$ estimated with lugsail BM estimators.  Although the true value of $\widehat{\text{ESS}}/n$ is not known in this case, the running plot mimics the previous example. Specifically, all four estimators seem to be converging to the true quantity, but the over lugsail seems to be converging more safely from below as opposed to the riskier convergence from above. This feature has a direct impact on the quality of inference. The methods converging from above would terminate the Markov chain earlier, yielding a false sense of security about the quality of estimation of the posterior mean.
\begin{figure}[htb]
	\centering
	\includegraphics[width = 2.5in]{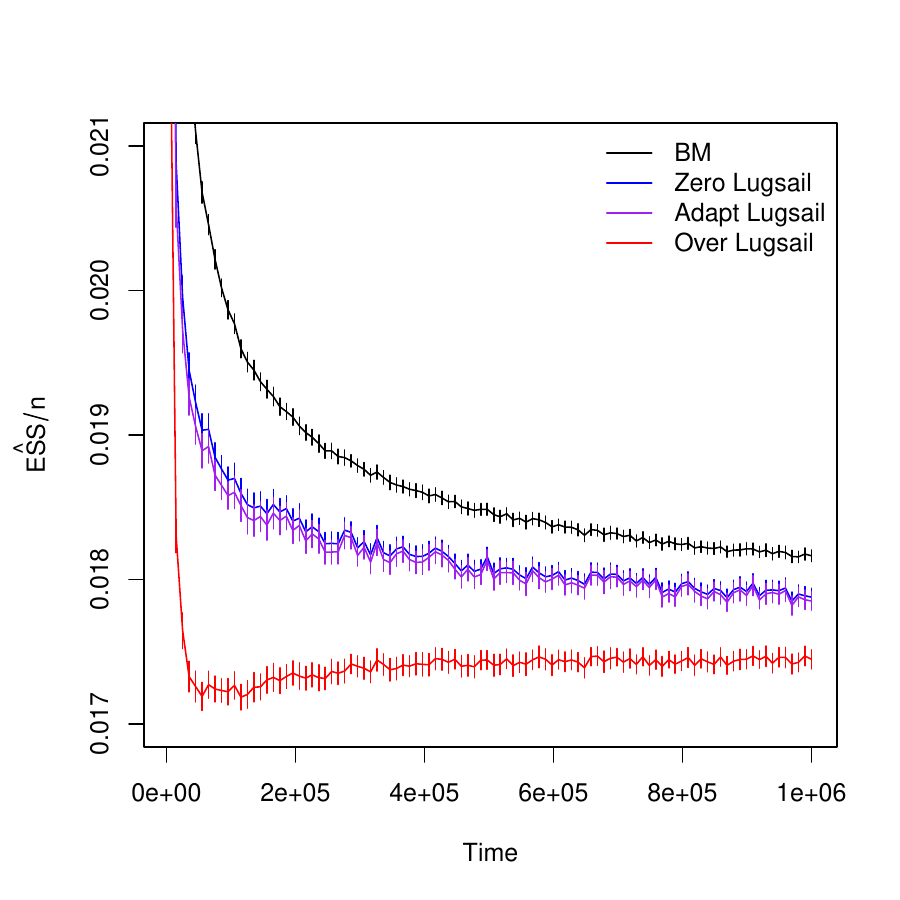}
	\caption{Logistic: Running plot of $\widehat{\text{ESS}}/n$ (with standard errors).}
	\label{fig:pro.boxplots}
\end{figure}

\section{Discussion} 
\label{sec:discussion}

We consider lugsail versions of the Bartlett, TH, and QS lag windows, but lugsail versions of other lag windows are readily available \cite[see][for an incomplete list]{ande:1971}. Lugsail lag windows can also be combined with other bias adjusting methods from \cite{kief:vogel:2002a}. We focus our attention on the family of consistent nonparametric estimators, however, see e.g.\ \cite{haan:1996,haan:levin:2000,mull:2007,mull:2014} for discussion on inconsistent and parametric estimators. 


For fixed $b$, our theoretical results imply a partial increase in the variance of the lugsail estimators, yielding an increase in the mean-squared error. However, there is no reason for lugsail and non-lugsail estimators to use the same $b$. Empirically, we find that lugsail estimators typically require a smaller $b$, implying a decrease in variability of the estimator. Indeed finding optimal choices of $b$ is a rich avenue for future work.

We believe this is the first instance of a lag window that takes values above 1. Our choice of method for creating these lugsail family of lag windows is motivated by the specific applications considered here.  Lag windows are also useful in signal-processing for spectral analysis and the estimation of instantaneous frequency \citep{boashash1992estimating}, and in genomic signal-processing \citep{gunawan2008optimal}. Our results  indicate that  application-specific developments of other lag windows taking values above 1 could yield similar finite-sample improvements. To facilitate future work in this area, reproducible codes for all examples and plots are available at \texttt{https://github.com/dvats/LugsailPaperCode}.

\section*{Acknowledgements} 
\label{sec:acknowledgements}
 The authors thank Daniel Eck, Karl Oskar Ekvall, and Galin Jones for critical feedback that improved the quality of presentation. Dootika Vats is supported by DST-SERB grant SPG/2021/001322.


\appendix
\section*{Appendix} 
\label{sec:appendix}


\section{Lag window calculations} 
\label{sec:different_lag_windows}
We present the bias  and variance expressions for the lugsail versions of spectral variance (SV) estimators.
\begin{enumerate}[(a)]
  \item \textit{Bartlett}. \cite{hannan:1970} shows that  $q = 1$ and $k_1 = 1$, so that the bias here is $\Gamma/b$ and $\int_{-\infty}^{\infty} k^2(x) = 2/3$. The first-order bias term for the lugsail Bartlett estimator is
  \[
  \dfrac{\Gamma}{b} \left( \dfrac{1 - rc_n}{1-c_n} \right)\,.
  \]
  Further,
  \[
  \int_{-\infty}^{\infty}k_L^2(x) = \dfrac{2}{3(1-c_n)^2}\left( 1 + \dfrac{c_n^2}{r}  - \dfrac{3c_n}{r} + \dfrac{c_n}{r^2}\right)\,.
  \]
  \item \textit{TH}. \cite{hannan:1970} shows that $q = 2$ and $k_q = \pi^2/4$. The TH estimator has a first-order bias term of $\Gamma^{(2)} \pi^2/(4 b^2) $. The lugsail TH estimator has bias
  \[
   \dfrac{\pi^2}{4} \dfrac{\Gamma^{(2)}}{b^2} \left( \dfrac{1 - c_nr^2}{1-c_n} \right)\,.
  \]
  For $r = 1$, $\int_{-\infty}^{\infty}k_L^2(x) = 3/4$ and for $r > 1$
  \[
  \int_{-\infty}^{\infty}k_L^2(x) = \dfrac{3}{4(1-c_n)^2}\left[ 1 + \dfrac{c_n^2}{r}  - \dfrac{4c_n}{3} \left(\dfrac{r^3 \sin (\pi/r) + \pi r^2 - \pi}{\pi r^3 - \pi r} \right)\right]\,.
  \]

  \item \textit{QS}. \cite{andr:1991} shows that $q = 2$ and $k_q = 1.4212$. The QS estimator has a first-order bias term of $\Gamma^{(2)} 1.4212/b^2 $, The lugsail QS  estimator has bias
  \[
   1.4212 \dfrac{\Gamma^{(2)}}{b^2} \left( \dfrac{1 - c_nr^2}{1-c_n} \right)\,.
  \]
  For $r = 1$, $\int_{\infty}^{\infty} k_L(x) = 1$, and for $r > 1$
  \[
  \int_{-\infty}^{\infty}k_L^2(x) = \dfrac{1}{(1-c_n)^2}\left\{ 1 + \dfrac{c_n^2}{r}  -  c_n\left[ \dfrac{1}{4r^3} \left((r-1)^3 (1 + 3r + r^2)  - (r+1)^3 (1-3r + r^2) \right) \right] \right\}\,.
  \]
\end{enumerate}

Table~\ref{tab:SV_vars} presents the variances of our recommended choices of lugsail parameter values for the three lag windows considered here. 
\begin{table}[htbp]
  \caption{Variances of SV estimators for lugsail parameters (rounded to 3 significant units). Here $a = b/n$.}
  \label{tab:SV_vars}
  \centering
 \begin{tabular}{|l|ccccc|}
  \hline
 Window & Original & Zero & Adapt  & Adapt  & Over\\
  & & & $(a=10^2)$ & $(a=10^5)$ & \\
  \hline
BT  & 0.667  & 1.333  & 1.522  & 1.407  & 1.704 \\ 
TH  & 0.750  & 0.964  & 1.017  & 0.985  & 0.986 \\ 
QS  & 1.000  & 1.306  & 1.381  & 1.335  & 1.329 \\ 
  \hline
  \end{tabular}
\end{table}

\section{Proof of Theorem~\ref{thm:bias_L}}
\label{sec:bias_proof}

Let $\alpha_g(n)$ be the mixing coefficient of $\{Y_t\}$ and let $\alpha(n)$ be the mixing coefficient of $\{X_t\}$. By \cite{doss:fleg:jon:2014}, $\alpha_g(n) \leq \alpha(n)$. Due to \citet[Theorem 17.2.2]{ibra:linn:1971} there exists $W <\infty$ which depends on the moments of $Y^{(i)}$ and $Y^{(j)}$ such that,
\[
\left|\text{Cov}_F\left(Y_1^{(i)}, Y_{1+s}^{(j)}\right)\right| \leq W \alpha_g(s)^{\delta/(2+\delta)} \leq W' s^{-(2+\epsilon)} \,.
\]
First, we show that under Assumption~\ref{ass:alpha_mixing} with $q = 1$, $|\Gamma_{ij}| < \infty$. Consider
\begin{align*}
|\Gamma_{ij}| & = \left|\sum_{s=1}^{\infty} s\left[\text{Cov}_F\left(Y_1^{(i)}, Y_{1+s}^{(j)}\right) + \text{Cov}_F\left(Y_1^{(j)}, Y_{1+s}^{(i)}\right) \right] \right|\\ 
& \leq \sum_{s=1}^{\infty} s\left[\left|\text{Cov}_F\left(Y_1^{(i)}, Y_{1+s}^{(j)} \right) \right| +  \left|\text{Cov}_F \left(Y_1^{(j)}, Y_{1+s}^{(i)} \right) \right| \right] \\ 
& \leq \sum_{s=1}^{\infty} 2W's \cdot s^{-2-\epsilon} < \infty\,.
\end{align*}
Let $\left[\text{Var}_F{(\bar{Y})}\right]_{ij}$ be the $ij$th element of $\Var_F{(\bar{Y})}$. By \citet[Lemma 1-3]{song:schm:1995}, 
\begin{equation}
\label{eq:nVar}
  n\left[\Var_F{(\bar{Y})}\right]_{ij}  = \Sigma_{ij} + \dfrac{\Gamma_{ij}}{n} + o(n^{-1})\,.
\end{equation}
Let $\hat{\Sigma}^{ij}_b$ be the $ij$th element of $\hat{\Sigma}_b$. Using \eqref{eq:nVar},
  \begin{align*}
    \E_F\left[ \hat{\Sigma}_b^{ij} \right] & = \E_F \left[\dfrac{b}{a-1} \sum_{l=0}^{a-1} \left(\bar{Y}^{(i)}_l(b) - \bar{Y}^{(i)} \right) \left(\bar{Y}^{(j)}_l(b) - \bar{Y}^{(j)} \right)  \right]\\ 
  & = \dfrac{ab}{a-1} \left(\text{Cov}_F \left(\bar{Y}_1(b)^{(i)}, \bar{Y}_1(b)^{(j)} \right)- \text{Cov}_F \left(\bar{Y}^{(i)}, \bar{Y}^{(j)} \right)\right)\\
  & = \dfrac{ab}{a-1} \left(\Sigma_{ij}\dfrac{a-1}{ab} + \Gamma_{ij} \dfrac{(a-1)(a+1)}{a^2b^2}  + o \left(b^{-1} \right) \right)\\
  & = \Sigma_{ij} + \dfrac{\Gamma_{ij}}{b} + o \left(\dfrac{1}{b} \right)\,.
  \end{align*}

\section{Proof of Theorem~\ref{thm:variance}}
\label{sec:proof_of_multivariate_bm}
\begin{proposition}
\label{prop:exp_product_squares}
  If $(X,Y)$ is a mean 0 bivariate normal random variable such that
  \[ 
  \left[\begin{array}{c} X \\ Y\end{array} \right] \sim N \left(\left[\begin{array}{c} 0 \\ 0\end{array} \right] \,,\, \left[\begin{array}{cc} l_{11} & l_{12} \\ l_{12} & l_{22} \end{array} \right]   \right)\,, 
  \] 
  then $\E[X^2Y^2] = 2l_{12}^2 + l_{11} l_{22}$. 
\end{proposition}

\smallskip
\begin{proposition}
\label{prop:4th_product_expectation}
\citep{jans:stoi:1988} If $(X_1, X_2, X_3,X_4)$ is mean 0 normally distributed, then $\E[X_1X_2X_3X_4] = \E[X_1X_2]\E[X_3X_4]$ $+ \E[X_1X_3]\E[X_2X_4] + \E[X_1X_4]\E[X_2X_3]$.
\end{proposition}

\smallskip
Recall that $B(t)$ is a $p$-dimensional standard Brownian motion. Let $B^{(i)}(t)$ denote the $i$th component of the vector $B(t)$ and let $\bar{B} = n^{-1}B(n)$, and $\bar{B}_l(s) = s^{-1}[B(ls+s) - B(ls)]$. Recall $\Sigma = LL^T$ where $L$ is the lower triangular matrix in \eqref{eq:SIP}. Define the $p$-dimensional scaled Brownian motion $C(t) = LB(t)$ and let $C^{(i)}(t)$ be the $i$th component of $C(t)$. In addition, define $\bar{C}^{(i)} = n^{-1} C^{(i)}(n)$ and $\bar{C}^{(i)}_l(s) = s^{-1}[C^{(i)}(ls+s) - C^{(i)}(ls)]$. 

Consider the Brownian motion equivalent of the lugsail BM estimator,
\begin{align*}
\tilde{\Sigma}_{L} & = \dfrac{1}{1-c_n} \dfrac{b}{a - 1} \ds \sum_{l=0}^{a - 1} (\bar{B}_l(b) - \bar{B})(\bar{B}_l(b) - \bar{B})^T \\
& \quad  - \dfrac{c_n}{1-c_n} \dfrac{b/r}{ra - 1} \ds \sum_{l=0}^{ra - 1} (\bar{B}_l(b/r) - \bar{B})(\bar{B}_l(b/r) - \bar{B})^T\,. 
\end{align*}
We will show that the variance of the lugsail BM estimator is the same as the variance of the Brownian motion equivalent. The following lemma will be needed later in the proof.
\begin{lemma}
\label{lem:tilde_hat_equi_moment}
Suppose Assumption~\ref{ass:b_increasing} holds and the assumptions of Theorem~\ref{thm:kuelbs} holds with $D$ such that $\E D^4 <\infty$ and $\psi(n)^2 b^{-1} \log n  \to 0$ as $n \to \infty$, then as $n \to \infty$,
\[
 \E \left[ \left(\hat{\Sigma}^{ij}_L - \Sigma_{ij}\tilde{\Sigma}^{ij}_{L}  \right)^2 \right] \to 0\,.
\]
\end{lemma}

\begin{proof}
Under the same conditions, \cite{vats:fleg:jones:2017b} showed that, as $n\to \infty$
\[
\left|\Sigma^{ij}_b - \Sigma_{ij}\tilde{\Sigma}^{ij}_{b}  \right| \overset{a.s.}{\to} 0.
\]
Note that,
\begin{align*}
\label{eq:lug_abs_diff_0}
\left|\hat{\Sigma}^{ij}_L - \Sigma_{ij}\tilde{\Sigma}^{ij}_{L}  \right| & = \left|\dfrac{1}{1-c_n} \left(\Sigma^{ij}_b - \Sigma_{ij}\tilde{\Sigma}^{ij}_{b} \right) - \dfrac{c_n}{1-c_n} \left(\Sigma^{ij}_{b/r} - \Sigma_{ij}\tilde{\Sigma}^{ij}_{b/r} \right) \right| \\ 
& \leq \left|\dfrac{1}{1-c_n} \left(\Sigma^{ij}_b - \Sigma_{ij}\tilde{\Sigma}^{ij}_{b} \right) \right| + \left| \dfrac{c_n}{1-c_n} \left(\Sigma^{ij}_{b/r} - \Sigma_{ij}\tilde{\Sigma}^{ij}_{b/r} \right) \right| \\ \numberthis
& \to 0 \text{ with probability 1 as } n \to \infty\,.
\end{align*}
Define,
\[
g_1(b,n) = 4 \dfrac{a}{a-1} \dfrac{\log n}{b} \psi(n)^2, \quad \text{ and for some }  \epsilon > 0
\]
\[
g_2(b,n) = 4 \sqrt{\Sigma_{ii}} (1+\epsilon) \dfrac{a}{a-1} \left( \dfrac{\log n}{b} \psi(n)^2 \right)^{1/2} \,.
\]
If $\psi(n)^2 b^{-1} \log n  \to 0$ as $n \to \infty$, then $g_1(b,n) \to 0$ and $g_2(b,n) \to 0$ as $n \to \infty$. 
From \cite{vats:fleg:jones:2017b}, for the $(i,j)$the element of $\Sigma_b$,
\begin{equation*}
\label{eq:g1g2}
|\Sigma^{ij}_b  - \Sigma_{ij}\tilde{\Sigma}_b^{ij}| \leq D^2 g_1(n) + Dg_2(n)\,,
\end{equation*}
where $D$ is the random variable in the conclusion of Theorem~\ref{thm:kuelbs}. For the lugsail BM estimator,
\begin{align*}
|\hat{\Sigma}^{ij}_L  - \Sigma_{ij} \tilde{\Sigma}_L^{ij} | &\leq \left|\dfrac{1}{1-c_n} \left(\Sigma^{ij}_b - \Sigma_{ij}\tilde{\Sigma}^{ij}_{b} \right) \right| + \left| \dfrac{c_n}{1-c_n} \left(\Sigma^{ij}_{b/r} - \Sigma_{ij}\tilde{\Sigma}^{ij}_{b/r} \right) \right| \\ 
& \leq \dfrac{1}{1-c_n} \left(D^2 g_1(b,n) + Dg_2(b,n) \right) + \dfrac{c_n}{1-c_n} \left(D^2 g_1(b/r,n) + Dg_2(b/r, n) \right)\\
& \leq D^2 g_1^*(n) + Dg_2^*(n)\,
\end{align*}
where
\[
g_1^*(n) = \dfrac{1}{1-c_n} \left(g_1(b,n) +  g_1(b/r, n) \right) \quad \text{ and } \quad 
 g_2^*(n) = \dfrac{1}{1-c_n} \left(g_2(b,n) +  g_2(b/r, n) \right)\,.
\]
If $\psi(n)^2 b^{-1} \log n  \to 0$ then $g_1^*(n) \to 0$, $g_2^*(n) \to 0$.  %
By \eqref{eq:lug_abs_diff_0}, there exists integer $N_0$ such that
\begin{align*}
\left(\hat{\Sigma}^{ij}_L  - \Sigma_{ij}\tilde{\Sigma}_L^{ij} \right)^2 & = \left(\hat{\Sigma}^{ij}_L  - \Sigma_{ij}\tilde{\Sigma}_L^{ij} \right)^2 I(0 \leq n \leq N_0) + \left(\hat{\Sigma}^{ij}_L  - \Sigma_{ij}\tilde{\Sigma}_L^{ij}  \right)^2 I(N_0 < n) \\
& \leq \left(\hat{\Sigma}^{ij}_L  - \Sigma_{ij}\tilde{\Sigma}_L^{ij} \right)^2 I(0 \leq n \leq N_0) +  \left( D^2 g_1^*(n) + Dg_2^*(n) \right)^2I(N_0 < n) \\ 
& = g_n^*(Y_1, \dots, Y_n, B(0), \dots, B(n))\,.
\end{align*}
Since $\E D^4 <\infty, \E \left(\hat{\Sigma}^{ij}_L\right)^2 < \infty$, and $\E \left(\tilde{\Sigma}^{ij}_L \right)^2 < \infty$,
\begin{align*}
  \E \left| g_n^*(Y_1, \dots, Y_n, B(0), \dots, B(n)) \right| &\leq \E \left(\hat{\Sigma}^{ij}_L  - \Sigma_{ij}\tilde{\Sigma}_L^{ij} \right)^2 + \E \left(D^2g_1^*(n) + D g_2^*(n) \right)^2 <\infty
\end{align*}
Thus, $\E|g_n| <\infty$, $g_n \to 0$ as $n\to \infty$ with probability 1, and $\E g_n \to 0$ as $n \to \infty$. Since $\left(\hat{\Sigma}^{ij}_L  - \Sigma_{ij}\tilde{\Sigma}_L^{ij} \right)^2  \to 0$ with probability 1, the generalized majorized convergence theorem \citep{zeidler2013nonlinear} yields 
\[
 \E \left[ \left(\hat{\Sigma}^{ij}_L - \Sigma_{ij}\tilde{\Sigma}^{ij}_{L}  \right)^2 \right] \to 0\,.
\]
\end{proof}

\begin{proof}[of Theorem~\ref{thm:variance}]
We will prove Theorem~\ref{thm:variance} by first finding the variance of $\tilde{\Sigma}_{L,L}$ and then using Lemma~\ref{lem:tilde_hat_equi_moment} show that this is equal to the variance of the lugsail BM estimator. Let $\tilde{\Sigma}_{L,L} = L \tilde{\Sigma}_L L^T$, so that
\begin{align*}
 \tilde{\Sigma}_{L,L} 
& = \dfrac{1}{1-c_n} \dfrac{b}{a - 1} \ds \sum_{l=0}^{a - 1} (\bar{C}_l(b) - \bar{C})(\bar{C}_l(b) - \bar{C})^T \\
& \quad - \dfrac{c_n}{1-c_n} \dfrac{b/r}{ra - 1} \ds \sum_{l=0}^{ra - 1} (\bar{C}_l(b/r) - \bar{C})(\bar{C}_l(b/r) - \bar{C})^T\,.  
\end{align*}
First we establish some useful identities. Let $U_t^{(i)} = B^{(i)}(t) - B^{(i)}(t-1)$, then $U^{(i)}_t \overset{iid}{\sim} N(0, 1)$ for $t = 1, 2, \dots$. In addition, for any batch $l$, and index $s$
\[
\bar{B}_l^{(i)}(s) - \bar{B} = \left(\dfrac{n-s}{ns} \right)\sum_{t = l}^{l+s} U_t^{(i)} - \dfrac{1}{n}\sum_{t=1}^{l} U^{(i)}_t - \dfrac{1}{n} \sum_{t = l+s+1}^{n} U^{(i)}_t\,.
\]
Since $\E[\bar{B}^{(i)}_l(s) - \bar{B}^{(i)}] = 0$ for $l = 0, \dots, n-s$,
\[
\text{Var}[\bar{B}^{(i)}_l(s) - \bar{B}^{(i)}] = \left( \dfrac{n-s}{ns} \right)^2s + \dfrac{n-s}{n^2} = \dfrac{n-s}{ns}\,.
\]
Since $B$ is a $p$-dimensional standard Brownian motion and $C(t) = LB(t)$,
\begin{equation}
  \label{eq:centered_C_dist}
 \bar{C}_l(s) - \bar{C} \sim N \left(0, \dfrac{n-s}{sn} \Sigma \right)\,.
\end{equation} 
Also for $s \geq b$ in \citet[Equation 19]{liu:vats:fle:2018},
\begin{equation}
  \label{eq:cov_C_lags}
\text{Cov}(\bar{C}_{l}(b) - \bar{C}, \bar{C}_{l+s}(b) - \bar{C}) = -\dfrac{\Sigma}{n}\,.
\end{equation}
In addition, for $p = 0, \dots, a-1$, and $q = rp, rp+1, \dots, r(p+1) - 1$,
\begin{equation}
  \label{eq:Cov_C_overlap}
  \text{Cov}\left(\bar{C}_p(b) - \bar{C}, \bar{C}_q(b/r) - \bar{C}  \right)  = \dfrac{n-b}{bn} \Sigma\,\,,
\end{equation}
and for $p = 0, \dots, a-1$ and $q \ne rp, rp+1, \dots, r(p+1) - 1$
\begin{equation}
  \label{eq:Cov_C_no_over}
  \text{Cov}\left(\bar{C}_p(b) - \bar{C}, \bar{C}_q(b/r) - \bar{C} \right)  =  -\dfrac{\Sigma}{n}\,.
\end{equation}
We will consider the variance of each individual term of $\tilde{\Sigma}_{L,L}$. Note that $\Var\left[\tilde{\Sigma}_{L,L}^{ij}\right] = \E\left[\tilde{\Sigma}_{L,L}^{2,ij} \right] - \left(\E[\tilde{\Sigma}_{L,L}^{ij}] \right)^2$. Consider,
\begin{align*}
\E\left[\tilde{\Sigma}_{L,L}^{2,ij} \right] & = \E \Bigg[\Bigg( \dfrac{1}{1-c_n} \dfrac{b}{a - 1} \ds \sum_{l=0}^{a - 1} \left(\bar{C}^{(i)}_l(b) - \bar{C}^{(i)} \right) \left(\bar{C}^{(j)}_l(b) - \bar{C}^{(j)} \right)\\ 
& \quad  - \dfrac{c_n}{1-c_n} \dfrac{b/r}{ra - 1} \ds \sum_{l=0}^{ra - 1} \left(\bar{C}^{(i)}_l(b/r) - \bar{C}^{(i)} \right) \left(\bar{C}^{(j)}_l(b/r) - \bar{C}^{(j)} \right)  \Bigg)^2\Bigg]\\
& = \E \Bigg[ \left( \dfrac{1}{1-c_n} \right)^2 \left( \dfrac{b}{a-1} \right)^2 \left(  \ds \sum_{l=0}^{a - 1} \left(\bar{C}^{(i)}_l(b) - \bar{C}^{(i)} \right)  \left(\bar{C}^{(j)}_l(b) - \bar{C}^{(j)} \right) \right)^2\\
&\quad +  \left( \dfrac{c_n}{1-c_n} \right)^2 \left(  \dfrac{b/r}{ra-1}\right)^2 \left( \ds \sum_{l=0}^{ra - 1} \left(\bar{C}^{(i)}_l(b/r) - \bar{C}^{(i)} \right)  \left(\bar{C}^{(j)}_l(b/r) - \bar{C}^{(j)} \right)\right)^2\\
& \quad  - \dfrac{2c_n}{(1-c_n)^2} \dfrac{b^2/r}{(a-1)(ra-1)}     \left(  \ds \sum_{l=0}^{a - 1} \left(\bar{C}^{(i)}_l(b) - \bar{C}^{(i)} \right)  \left(\bar{C}^{(j)}_l(b) - \bar{C}^{(j)} \right) \right) \\
& \qquad \times  \left( \ds \sum_{l=0}^{ra - 1} \left(\bar{C}^{(i)}_l(b/r) - \bar{C}^{(i)} \right)  \left(\bar{C}^{(j)}_l(b/r) - \bar{C}^{(j)} \right)  \right) \Bigg]\\
& = A_1 + A_2 + A_3\,, \numberthis \label{eq:4thmoment_first}
\end{align*}  
where
\begin{align*}
A_1 & = \E \Bigg[ \left( \dfrac{1}{1-c_n} \right)^2 \left( \dfrac{b}{a-1} \right)^2 \left(  \ds \sum_{l=0}^{a - 1} \left(\bar{C}^{(i)}_l(b) - \bar{C}^{(i)}\right)  \left(\bar{C}^{(j)}_l(b) - \bar{C}^{(j)} \right) \right)^2 \Bigg]\,,\\
A_2 & =   \E \Bigg[  \left( \dfrac{c_n}{1-c_n} \right)^2 \left(  \dfrac{b/r}{ra-1}\right)^2 \left( \ds \sum_{l=0}^{ra - 1} \left(\bar{C}^{(i)}_l(b/r) - \bar{C}^{(i)} \right)  \left(\bar{C}^{(j)}_l(b/r) - \bar{C}^{(j)} \right)\right)^2 \Bigg]\,, \text{ and }\\
A_3 & = \E \Bigg[ - \dfrac{2c_n}{(1-c_n)^2} \dfrac{b^2/r}{(a-1)(ra-1)}     \left(  \ds \sum_{l=0}^{a - 1} \left(\bar{C}^{(i)}_l(b) - \bar{C}^{(i)} \right)  \left(\bar{C}^{(j)}_l(b) - \bar{C}^{(j)} \right) \right) \\
& \qquad \times  \left( \ds \sum_{l=0}^{ra - 1} \left(\bar{C}^{(i)}_l(b/r) - \bar{C}^{(i)} \right)  \left(\bar{C}^{(j)}_l(b/r) - \bar{C}^{(j)} \right)\right) \Bigg]\,.
\end{align*}
Consider $A_1$,
\begin{align*}
A_1 
& = \left( \dfrac{1}{1-c_n} \right)^2 \left( \dfrac{b}{a-1} \right)^2 \E \Bigg[ \ds \sum_{l=0}^{a - 1} \left(\bar{C}^{(i)}_l(b) - \bar{C}^{(i)} \right)^2  \left(\bar{C}^{(j)}_l(b) - \bar{C}^{(j)} \right)^2    + \\
& \quad + 2 \sum_{s=1}^{a-1} \sum_{l=0}^{a-1-s} \left(\bar{C}^{(i)}_l(b) - \bar{C}^{(i)} \right)  \left(\bar{C}^{(j)}_l(b) - \bar{C}^{(j)} \right)   \left(\bar{C}^{(i)}_{l+s}(b) - \bar{C}^{(i)} \right)  \left(\bar{C}^{(j)}_{l+s}(b) - \bar{C}^{(j)} \right)  \Bigg]\\
& = \left( \dfrac{1}{1-c_n} \right)^2 \left( \dfrac{b}{a-1} \right)^2  \E[a_1 + 2a_2]\,, \numberthis \label{eq:A1_lugsail_bm}
\end{align*}
where 
\begin{align*}
  a_1 & = \ds \sum_{l=0}^{a - 1} \left(\bar{C}^{(i)}_l(b) - \bar{C}^{(i)} \right)^2  \left(\bar{C}^{(j)}_l(b) - \bar{C}^{(j)} \right)^2\,,\\
  a_2 & = \sum_{s=1}^{a-1} \sum_{l=0}^{a-1-s} \left(\bar{C}^{(i)}_l(b) - \bar{C}^{(i)} \right)  \left(\bar{C}^{(j)}_l(b) - \bar{C}^{(j)} \right)  \left(\bar{C}^{(i)}_{l+s}(b/r) - \bar{C}^{(i)} \right)  \left(\bar{C}^{(j)}_{l+s}(b/r) - \bar{C}^{(j)} \right)\,.
\end{align*}
We first consider the $a_1$ term. Using Proposition~\ref{prop:exp_product_squares} and \eqref{eq:centered_C_dist} with $k = b$, 
\begin{align*}
 \E \left[ \left(\bar{C}^{(i)}_l(b) - \bar{C}^{(i)} \right)^2  \left(\bar{C}^{(j)}_l(b) - \bar{C}^{(j)} \right)^2 \right] = \left( \dfrac{n-b}{bn}\right)^2[2\Sigma^2_{ij} + \Sigma_{ii} \Sigma_{jj}]\,.\numberthis \label{eq:Cijsquare}
\end{align*}
Using \eqref{eq:Cijsquare} for $\E[a_1]$,
\begin{align*}
\E[a_1] 
& = a  \left[\left( \dfrac{n-b}{bn}\right)^2[2 \Sigma^2_{ij} + \Sigma_{ij} \Sigma_{jj}] \right]\\
& = a  \left[\left(  \dfrac{1}{b^2} + \dfrac{1}{n^2} - \dfrac{2}{bn}\right) \left[2 \Sigma^2_{ij} + \Sigma_{ii} \Sigma_{jj} \right] \right] = \dfrac{a}{b^2} [2 \Sigma^2_{ij} + \Sigma_{ii} \Sigma_{jj}] + o \left( \dfrac{a}{b^2} \right)\,. \numberthis \label{eq:a_1_lugsail_bm}
\end{align*}
For the term $a_2$, using Proposition~\ref{prop:4th_product_expectation}, \eqref{eq:centered_C_dist}, and \eqref{eq:cov_C_lags} with $s = b$,
\begin{align*}
& \E \left[   \left(\bar{C}^{(i)}_l(b) - \bar{C}^{(i)} \right)  \left(\bar{C}^{(j)}_l(b) - \bar{C}^{(j)} \right)   \left(\bar{C}^{(i)}_{l+s}(b) - \bar{C}^{(i)} \right)  \left(\bar{C}^{(j)}_{l+s}(b) - \bar{C}^{(j)} \right)  \right]\\
& = \left[ \left( \dfrac{n-b}{bn} \right)^2 \Sigma^2_{ij}  + \dfrac{\Sigma_{ii} \Sigma_{jj} }{n^2}  + \dfrac{\Sigma^2_{ij} }{n^2} \right] =  \left[ \dfrac{\Sigma_{ij}^2}{b^2} + \dfrac{1}{n^2}[2 \Sigma_{ij}^2 + \Sigma_{ii} \Sigma_{jj}] - \dfrac{2\Sigma_{ij}^2}{bn}  \right]\,. \numberthis \label{eq:4productsCij}
\end{align*}
Using \eqref{eq:4productsCij} in calculating $\E[a_2]$,
\begin{align*}
  \E[a_2] & = \sum_{s=1}^{a-1} \sum_{l=0}^{a-1-s} \left(  \dfrac{\Sigma_{ij}^2}{b^2} + \dfrac{1}{n^2}[2 \Sigma_{ij}^2 + \Sigma_{ii} \Sigma_{jj}] - \dfrac{2\Sigma_{ij}^2}{bn}   \right)\\
  & = \dfrac{a(a-1)}{2} \left(  \dfrac{\Sigma_{ij}^2}{b^2} + \dfrac{1}{n^2}[2 \Sigma_{ij}^2 + \Sigma_{ii} \Sigma_{jj}] - \dfrac{2\Sigma_{ij}^2}{bn}   \right)\\
  & = \dfrac{\Sigma^2_{ij}}{2} \left[ \dfrac{(a-1)(a - 2)}{b^2} \right] + o\left( \dfrac{a}{b^2} \right)\,. \numberthis \label{eq:a_2_lugsail_bm}
\end{align*}
Using \eqref{eq:a_1_lugsail_bm} and \eqref{eq:a_2_lugsail_bm} in \eqref{eq:A1_lugsail_bm},
\begin{align*}
  \E[A_1] & = \left( \dfrac{1}{1-c_n} \right)^2 \left( \dfrac{b}{a-1} \right)^2\E[a_1 + 2a_2]\\
  & = \left( \dfrac{1}{1-c_n} \right)^2 \left( \dfrac{b}{a-1} \right)^2 \left[ \dfrac{a}{b^2} [2 \Sigma^2_{ij} + \Sigma_{ii} \Sigma_{jj}] + \Sigma^2_{ij} \left[ \dfrac{(a-1)(a - 2)}{b^2} \right] + o\left( \dfrac{a}{b^2} \right) \right]\\
  & = \dfrac{1}{(1-c_n)^2} \dfrac{ a(a-1)\Sigma_{ij}^2 + a \Sigma_{ii} \Sigma_{jj} }{(a-1)^2}  + o\left( \dfrac{1}{a} \right)\,. \numberthis \label{eq:A1_final_bm}
\end{align*}
Similarly for $A_2$, 
\begin{align*}
  \E[A_2] & = \left(\dfrac{c_n}{1-c_n}\right)^2 \dfrac{ ra(ra-1)\Sigma_{ij}^2 + ra \Sigma_{ii} \Sigma_{jj} }{(ra-1)^2}  + o\left( \dfrac{1}{a} \right)\,. \numberthis \label{eq:A2_final_bm}
\end{align*}
We move on to the final term $A_3$,
\begin{align*}
& A_3 \\
 & = - \dfrac{2c_n}{(1-c_n)^2} \dfrac{b^2/r}{(a-1)(ra-1)}    \\
 & \, \times \E \left\{  \left[  \ds \sum_{p=0}^{a - 1}\ds \sum_{q=rp}^{r(p+1)-1} \left(\bar{C}^{(i)}_l(b) - \bar{C}^{(i)}\right) \left(\bar{C}^{(j)}_l(b) - \bar{C}^{(j)} \right)  \left(\bar{C}^{(i)}_l(b/r) - \bar{C}^{(i)} \right)  \left(\bar{C}^{(j)}_l(b/r) - \bar{C}^{(j)} \right) \right]  \right.\\
& \, \left. \times   \left(  \ds \sum_{p=0}^{a - 1}\ds \sum_{q\ne rp}^{r(p+1)-1} \left[\bar{C}^{(i)}_l(b) - \bar{C}^{(i)} \right) \left(\bar{C}^{(j)}_l(b) - \bar{C}^{(j)} \right)  \left(\bar{C}^{(i)}_l(b/r) - \bar{C}^{(i)} \right) \left(\bar{C}^{(j)}_l(b/r) - \bar{C}^{(j)} \right) \right] \right\}\,.
\end{align*}
For $p = 0, \dots, a-1$ and $q = rp, rp+1 , \dots, r(p+1) - 1$, let,
\[
OL = \left(\bar{C}_p^{(i)}(b) - \bar{C}^{(i)} \right) \left(\bar{C}_q^{(j)}(b) - \bar{C}^{(j)} \right) \left(\bar{C}_p^{(i)}(b/r) - \bar{C}^{(i)} \right)  \left(\bar{C}_q^{(i)}(b/r) - \bar{C}^{(i)} \right)\,.
\]
By Proposition~\ref{prop:4th_product_expectation}, \eqref{eq:centered_C_dist}, \eqref{eq:Cov_C_overlap}
\begin{align*}
  & \E\left[ \ds \sum_{p=0}^{a - 1}\ds \sum_{q=rp}^{r(p+1)-1}OL \right]\\
  & = ra\E\left[ \left(\bar{C}_p^{(i)}(b) - \bar{C}^{(i)} \right)  \left(\bar{C}_q^{(j)}(b) - \bar{C}^{(j)} \right) \left(\bar{C}_p^{(i)}(b/r) - \bar{C}^{(i)}\right)   \left(\bar{C}_q^{(i)}(b/r) - \bar{C}^{(i)} \right) \right]\\ 
  & = ra \left[\left( \dfrac{n-b}{bn} \right)^2[\Sigma_{ij}^2 + \Sigma_{ii} \Sigma_{jj}] + \left( \dfrac{n-b}{bn} \right) \left( \dfrac{rn-b}{bn} \right)\Sigma_{ij}^2 \right]\\
  & = \dfrac{r}{ab^2} \left[ (a-1)^2 [\Sigma_{ij}^2 + \Sigma_{ii} \Sigma_{jj}] + (a-1)(ra-1) \Sigma_{ij}^2  \right]\,. \numberthis \label{eq:OL_A3_mbm}
\end{align*}
Similarly, by Proposition~\ref{prop:4th_product_expectation}, \eqref{eq:centered_C_dist}, \eqref{eq:Cov_C_no_over},
\begin{align*}
& \E \left[\left(  \ds \sum_{p=0}^{a - 1}\ds \sum_{q\ne rp}^{r(p+1)-1} \left(\bar{C}^{(i)}_l(b) - \bar{C}^{(i)}\right)  \left(\bar{C}^{(j)}_l(b) - \bar{C}^{(j)}\right) \left(\bar{C}^{(i)}_l(b/r) - \bar{C}^{(i)}\right)  \left(\bar{C}^{(j)}_l(b/r) - \bar{C}^{(j)} \right) \right) \right]\\
& = ra(a-1) \E \left[ \left(\bar{C}^{(i)}_l(b) - \bar{C}^{(i)} \right) \left(\bar{C}^{(j)}_l(b) - \bar{C}^{(j)} \right) \left(\bar{C}^{(i)}_l(b/r) - \bar{C}^{(i)} \right) \left(\bar{C}^{(j)}_l(b/r) - \bar{C}^{(j)} \right) \right]\\
& =  ra(a-1) \left[ \dfrac{\Sigma_{ij}^2 + \Sigma_{ii} \Sigma_{jj} }{n^2} + \left( \dfrac{n-b}{bn} \right) \left(\dfrac{rn - b}{bn}  \right) \Sigma_{ij}^2 \right]\\ 
& = \dfrac{r(a-1)}{ab^2} \left[(a-1)(ra-1) \Sigma_{ij}^2 + (\Sigma_{ij}^2 + \Sigma_{ii}\Sigma_{jj} )  \right]\,. \numberthis \label{eq:NOL_A3_mbm}
\end{align*}
Using \eqref{eq:OL_A3_mbm} and \eqref{eq:NOL_A3_mbm} in $A_3$, we get
\begin{align*}
\E[A_3] & = \E \Bigg[ - \dfrac{2c_n}{(1-c_n)^2} \dfrac{b^2/r}{(a-1)(ra-1)}     \left(  \ds \sum_{l=0}^{a - 1} (\bar{C}^{(i)}_l(b) - \bar{C}^{(i)})(\bar{C}^{(j)}_l(b) - \bar{C}^{(j)}) \right) \\
& \qquad \times  \left( \ds \sum_{l=0}^{ra - 1} (\bar{C}^{(i)}_l(b/r) - \bar{C}^{(i)})(\bar{C}^{(j)}_l(b/r) - \bar{C}^{(j)})\right) \Bigg]\\
& = - \dfrac{2c_n}{(1-c_n)^2} \dfrac{b^2/r}{(a-1)(ra-1)}  \Bigg[ \dfrac{r}{ab^2} \left[ (a-1)^2 [\Sigma_{ij}^2 + \Sigma_{ii} \Sigma_{jj}]+ (a-1)(ra-1) \Sigma_{ij}^2  \right]\\ 
& \quad  + \dfrac{r(a-1)}{ab^2} \left[(a-1)(ra-1) \Sigma_{ij}^2 + (\Sigma_{ij}^2 + \Sigma_{ii}\Sigma_{jj} )  \right] \Bigg] \\ 
& = -\dfrac{2c_n}{(a-c_n)^2} \dfrac{1}{a(ra-1)} \left[(a-1) [\Sigma_{ij}^2 + \Sigma_{ii}\Sigma_{jj}] + ra\Sigma_{ij}^2 + (a-1)(ra-1) \Sigma_{ij}^2 + \Sigma_{ii} \Sigma{jj}  \right]\\
& = -\dfrac{2c_n}{(1-c_n)^2} \dfrac{ar\Sigma_{ij}^2 + \Sigma_{ii}{\Sigma_{jj} } }{(ra-1)}\,. \numberthis \label{eq:A3_final_mbm}
\end{align*}
Combining \eqref{eq:A1_final_bm}, \eqref{eq:A2_final_bm}, and \eqref{eq:A3_final_mbm} in \eqref{eq:4thmoment_first},
\begin{align*}
& \E[\tilde{\Sigma}_{L,L}^{2,ij}]\\ 
& = \dfrac{1}{(1-c_n)^2} \dfrac{ a(a-1)\Sigma_{ij}^2 + a \Sigma_{ii} \Sigma_{jj} }{(a-1)^2}   + \left(\dfrac{c_n}{1-c_n}\right)^2 \dfrac{ ra(ra-1)\Sigma_{ij}^2 + ra \Sigma_{ii} \Sigma_{jj} }{(ra-1)^2}\\ 
& \quad  -\dfrac{2c_n}{(1-c_n)^2} \dfrac{ar\Sigma_{ij}^2 + \Sigma_{ii}{\Sigma_{jj} } }{(ra-1)} + o\left( \dfrac{1}{a} \right)\\ 
& = \dfrac{(ra-1)^2 [ a(a-1)\Sigma_{ij}^2 + a \Sigma_{ii} \Sigma_{jj}]   + c_n^2(a-1)^2 [ra(ra-1)\Sigma_{ij}^2 + ra \Sigma_{ii} \Sigma_{jj}]}{(1-c_n)^2(a-1)^2(ra-1)^2}  + o\left(\dfrac{1}{a} \right)\\ 
& \quad - \dfrac{ 2c_n(a-1)^2(ra-1)[ar\Sigma_{ij}^2 + \Sigma_{ii}{\Sigma_{jj} }]  }{(1-c_n)^2(a-1)^2(ra-1)^2} \\ 
& = \dfrac{a^4r^2 (1-c_n)^2\Sigma_{ij}^2   + a^3\Sigma_{ii}\Sigma_{jj} (r^2 + c_n^2r - c_nr)  +a^3 \Sigma_{ij}^2 (-r^2 - c_n^2r + 2c_nr + 4r^2c_n - 2r^2c_n^2 - 2r) }{(1-c_n)^2(a-1)^2 (ra-1)^2}\\ 
& \quad + o\left(\dfrac{1}{a} \right)\,. \numberthis \label{eq:4th_moment_final}
\end{align*}
Moving on to $\E[\tilde{\Sigma}^{ij}_{L,L}]$ and using \eqref{eq:centered_C_dist},
\begin{align*}
\E[\tilde{\Sigma}^{ij}_{L,L}] & = \E\Bigg[ \dfrac{1}{1-c_n} \dfrac{b}{a - 1} \ds \sum_{l=0}^{a - 1} (\bar{C}^{(i)}_l(b) - \bar{C}^{(i)})(\bar{C}^{(j)}_l(b) - \bar{C}^{(j)})\\ 
& \quad  - \dfrac{c_n}{1-c_n} \dfrac{b/r}{ra - 1} \ds \sum_{l=0}^{ra - 1} (\bar{C}^{(i)}_l(b/r) - \bar{C}^{(i)})(\bar{C}^{(j)}_l(b/r) - \bar{C}^{(j)})  \Bigg]\\ 
& = \dfrac{1}{1-c_n} \dfrac{b}{a-1}a \left( \dfrac{n-b}{bn} \right)\Sigma_{ij} - \dfrac{c_n}{1-c_n} \dfrac{b/r}{ra - 1} ra \left(\dfrac{rn-b}{bn} \right) \Sigma_{ij} = \Sigma_{ij}\,. \numberthis \label{eq:expectation_final}
\end{align*}

Using \eqref{eq:4th_moment_final} and \eqref{eq:expectation_final} we get,
\begin{align*}
  & \Var(\tilde{\Sigma}^{ij}_{L,L}) \\
& = \E[\tilde{\Sigma}^{2,ij}_{L,L}] - (\E[\tilde{\Sigma}^{i,j}_{L,L}])^2\\ 
& = \dfrac{a^4r^2 (1-c_n)^2\Sigma_{ij}^2   + a^3\Sigma_{ii}\Sigma_{jj} (r^2 + c_n^2r - c_nr)  +a^3 \Sigma_{ij}^2 (4r^2c_n - 2r^2c_n^2 - 2r-r^2 - c_n^2r + 2c_nr) }{(1-c_n)^2(a-1)^2(ra-1)^2}\\ 
& \quad  - \Sigma_{ij}^2 + o\left( \dfrac{1}{a} \right)+ o\left( \dfrac{1}{a} \right)\\ 
&  = \dfrac{a^3  \Sigma_{ii}\Sigma_{jj} (r^2 + c_n^2r - 2rc_n)}{(1-c_n)^2 (a-1)^2 (ra-1)^2}\\ 
& \quad + \dfrac{ a^3 \Sigma_{ij}^2(4r^2c_n - 2r^2c_n^2 - 2r - r^2 - c_n^2r + 2c_nr + 2r + 2r^2 + 2rc_n^2 + 2r^2c_n^2 -4rc_n - 4r^2c_n)}{(1-c_n)^2 (a-1)^2 (ra-1)^2} \\ 
& = \dfrac{a^3 (\Sigma_{ij}^2 +  \Sigma_{ii}\Sigma_{jj}) (r^2 + c_n^2r - 2rc_n)}{(1-c_n)^2 (a-1)^2 (ra-1)^2} + o \left(\dfrac{1}{a} \right)\\ 
& = \dfrac{ (\Sigma_{ij}^2 +  \Sigma_{ii}\Sigma_{jj}) (r^2 + c_n^2r - 2rc_n)}{(1-c_n)^2 r^2a} + o \left(\dfrac{1}{a} \right)\\ 
& = \left[\dfrac{1}{r}  + \dfrac{r-1}{r(1-c_n)^2}\right] \dfrac{(\Sigma_{ij}^2 +  \Sigma_{ii}\Sigma_{jj})}{a} + o \left(\dfrac{1}{a} \right)\,.
\end{align*}
To show that this is equivalent to the variance of the lugsail BM estimator, define
\[
\eta\left(\hat{\Sigma}^{ij}_L, \tilde{\Sigma}^{ij}_L \right) = \Var \left( \hat{\Sigma}^{ij}_L - \Sigma_{ij} \tilde{\Sigma}^{ij}_L  \right) + 2\Sigma_{ij} \E \left[ \left(\tilde{\Sigma}^{ij}_L - \E \tilde{\Sigma}^{ij}_L  \right) \left( \hat{\Sigma}^{ij}_L - \Sigma_{ij} \tilde{\Sigma}^{ij}_L  \right) \right]\,.
\]
Using Cauchy-Schwarz inequality and Lemma~\ref{lem:tilde_hat_equi_moment}
\begin{align*}
\left| \eta\left(\hat{\Sigma}^{ij}_L, \tilde{\Sigma}^{ij}_L \right)\right|  & = \left| \Var \left( \hat{\Sigma}^{ij}_L - \Sigma_{ij} \tilde{\Sigma}^{ij}_L  \right) + 2\Sigma_{ij} \E \left[ \left(\tilde{\Sigma}^{ij}_L - \E \tilde{\Sigma}^{ij}_L  \right) \left( \hat{\Sigma}^{ij}_L - \Sigma_{ij} \tilde{\Sigma}^{ij}_L  \right) \right]  \right| \\ 
& \leq  \left| \E \left[\left( \hat{\Sigma}^{ij}_L - \Sigma_{ij} \tilde{\Sigma}^{ij}_L  \right)^2 \right] + 2\Sigma_{ij}  \left(\E \left[\left( \hat{\Sigma}^{ij}_L - \Sigma_{ij} \tilde{\Sigma}^{ij}_L  \right)^2 \right] \Var\left(\tilde{\Sigma}^{ij}_L\right)  \right)^{1/2}  \right| \\ 
& = o(1) + 2\Sigma_{ij}\left(o(1) \left(O\left(\dfrac{b}{n} \right) + o\left(\dfrac{b}{n} \right) \right)  \right)^{1/2} = o(1)\,.
\end{align*}
Next,
\begin{align*}
\Var \left(\hat{\Sigma}^{ij}_{L} \right) & = \E\left[ \left(\hat{\Sigma}^{ij}_{L} - \E \left[ \hat{\Sigma}^{ij}_{L}\right] \right)^2  \right] \\ 
& = \E\left[ \left( \hat{\Sigma}^{ij}_{L} \pm \Sigma_{ij} \tilde{\Sigma}^{ij}_L \pm \Sigma_{ij}\E \left[\tilde{\Sigma}^{ij}_L \right] - \E \left[ \hat{\Sigma}^{ij}_{L}\right] \right)^2\right] \\
& = \E\left[ \left(   \left(\hat{\Sigma}^{ij}_{L} -  \Sigma_{ij} \tilde{\Sigma}^{ij}_L \right) + \left(\Sigma_{ij} \tilde{\Sigma}^{ij}_L  - \Sigma_{ij} \E \left[\tilde{\Sigma}^{ij}_L \right] \right) -  \left(\E \left[ \hat{\Sigma}^{ij}_{L}\right] - \Sigma_{ij} \E \left[\tilde{\Sigma}^{ij}_L \right] \right)\right)^2\right] \\
& =  \Sigma_{ij}^2 \E \left[ \left( \tilde{\Sigma}^{ij}_L  - \E \left[\tilde{\Sigma}^{ij}_L \right] \right)^2 \right]+ \E \left[ \left(\left(\hat{\Sigma}^{ij}_{L} -  \Sigma_{ij} \tilde{\Sigma}^{ij}_L \right) -  \left(\E \left[ \hat{\Sigma}^{ij}_{L}\right] - \Sigma_{ij} \E \left[\tilde{\Sigma}^{ij}_L \right] \right)  \right)^2\right] \\ 
& \quad \quad +  2\Sigma_{ij} \E \left[ \left(\tilde{\Sigma}^{ij}_L - \E \tilde{\Sigma}^{ij}_L  \right) \left[ \left( \hat{\Sigma}^{ij}_L - \Sigma_{ij} \tilde{\Sigma}^{ij}_L  \right) - \E \left(\hat{\Sigma}^{ij}_L - \Sigma_{ij}\tilde{\Sigma}^{ij}_L \right) \right] \right] \\ 
& = \Sigma_{ij}^2 \Var \left( \tilde{\Sigma}^{ij}_L\right) + \eta\left(\hat{\Sigma}^{ij}_L, \tilde{\Sigma}^{ij}_L \right) + o(1)\\
& = {(\Sigma_{ij}^2 +  \Sigma_{ii}\Sigma_{jj})}\left[\dfrac{1}{r}  + \dfrac{r-1}{r(1-c_n)^2}\right]\dfrac{b}{n}  + o \left(\dfrac{b}{n} \right) + o(1)\,.
\end{align*}
\end{proof}

\bibliographystyle{apalike}
\bibliography{mcref}

\end{document}